\documentclass[conference,compsoc]{IEEEtran}
\usepackage{cite}
\usepackage [
    n, 
    advantage,
    operators,
    sets,
    adversary,
    landau,
    probability,
    notions,
    logic,
    ff,
    mm,
    primitives,
    events,
    complexity,
    oracles,
    asymptotics,
    keys
]{cryptocode}
\usepackage{amsmath,amssymb,amsfonts}
\usepackage{algorithmic}
\usepackage{graphicx}
\usepackage{subfigure}
\usepackage{textcomp}
\usepackage{xcolor}
\usepackage{multicol}
\usepackage{multirow}
\usepackage{svg}
\usepackage{enumitem}
\usepackage{booktabs}
\usepackage{colortbl}
\usepackage{hyperref}
\usepackage{marvosym}

\newtheorem{definition}{Definition}
\newtheorem{theorem}{Theorem}
\newtheorem{lemma}{Lemma}
\newtheorem{proof}{Proof}
\newtheorem{corollary}{Corollary}

\ifCLASSINFOpdf
\else
\fi
\begin{document}
%
\title{Communication Efficient Multiparty Private Set Intersection from\\Multi-Point Sequential OPRF}





%
\author{\IEEEauthorblockN{Xinyu Feng\IEEEauthorrefmark{2}\IEEEauthorrefmark{3}\IEEEauthorrefmark{1},
Yukun Wang\IEEEauthorrefmark{2}\IEEEauthorrefmark{1},
Cong Li\IEEEauthorrefmark{4}\textsuperscript{\Letter},
Wu Xin\IEEEauthorrefmark{5},
Ming Yao\IEEEauthorrefmark{2}, 
Dian Zhang\IEEEauthorrefmark{2}, 
Wanwan Wang\IEEEauthorrefmark{2} and
Hao He\IEEEauthorrefmark{2}}
\IEEEauthorblockA{\IEEEauthorrefmark{2}InsightOne Tech Co., Ltd.}
\IEEEauthorblockA{\IEEEauthorrefmark{3}Zhonghe Tech (Xiong'an) Co., Ltd.}
\IEEEauthorblockA{\IEEEauthorrefmark{4}School of Software and Microelectronics, Peking University}
\IEEEauthorblockA{\IEEEauthorrefmark{5}School of Economics, Peking University}
\IEEEauthorblockA{\textsuperscript{\Letter}Corresponding Author}
\IEEEauthorblockA{\IEEEauthorrefmark{1}The two authors contributed equally to this work.}}


\maketitle

\begin{abstract}
Multiparty private set intersection (MPSI) allows multiple participants to compute the intersection of their locally owned data sets without revealing them. MPSI protocols can be categorized based on the network topology of nodes, with the star, mesh, and ring topologies being the primary types, respectively. Given that star and mesh topologies dominate current implementations, most existing MPSI protocols are based on these two topologies. However, star-topology MPSI protocols suffer from high leader node load, while mesh topology protocols suffer from high communication complexity and overhead. In this paper, we first propose a multi-point sequential oblivious pseudorandom function (MP-SOPRF) in a multi-party setting. Based on MP-SOPRF, we then develop an MPSI protocol with a ring topology, addressing the challenges of communication and computational overhead in existing protocols. We prove that our MPSI protocol is semi-honest secure under the Hamming correlation robustness assumption. Our experiments demonstrate that our MPSI protocol outperforms state-of-the-art protocols, achieving a reduction of 74.8\% in communication and a 6\% to 287\% improvement in computational efficiency.
\end{abstract}


%
\IEEEpeerreviewmaketitle

\section{Introduction}\label{sec:introduction}

Private set intersection (PSI) is a versatile cryptographic instrument that facilitates two participants, each with distinct local data sets, in determining the common elements of their datasets without exposing the nonintersecting components. As an extension of PSI, multiparty PSI (MPSI) enables multiple participants to ascertain their common data elements without revealing nonintersecting data, in a similar manner. MPSI has numerous practical applications across various domains. In particular, within vertical federated learning (VFL), MPSI is employed for data alignment based on the local data identifiers of each participant, serving as an essential prerequisite~\cite{lu2020multi}. In the domain of medical data collection, MPSI presents a promising technology for consolidating medical data from numerous disparate entities while ensuring the privacy of sensitive data~\cite{miyaji2017privacy}. Furthermore, in the field of cybersecurity protection, MPSI is instrumental for multiple organizations to collaboratively identify malicious intrusions on public networks without compromising the confidentiality of other users' information~\cite{kolesnikov2017practical, inbar2018efficient}.

\begin{figure}[htb]
    \centering
    \subfigure[Star Topology\label{fig:star-topology}]{\includegraphics[width=0.48\linewidth]{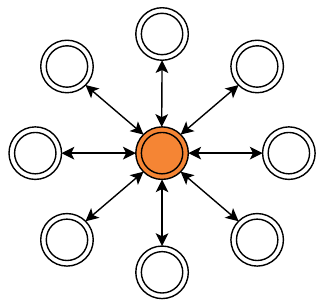}}\hspace{10pt}
    \subfigure[Mesh Topology]{\includegraphics[width=0.45\linewidth]{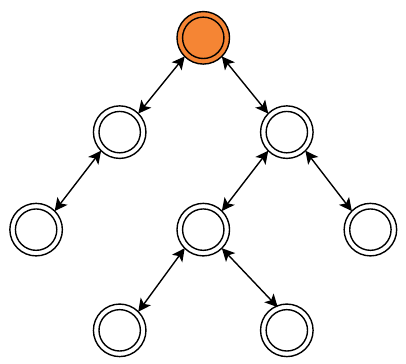}}\\
    \subfigure[Ring Topology]{\includegraphics[width=0.45\linewidth]{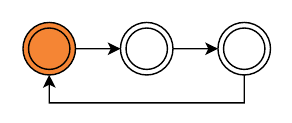}}
    \caption{Network Topologies of MPSI}
  \label{fig:network-topology}
\end{figure}

Within existing MPSI protocols, there are primarily three types of topological structures from a communication model perspective, namely, star topology, mesh topology, and ring topology \footnote{In various academic works, the ring topology is occasionally referred to as the ring topology~\cite{vos2023sok}.}, as depicted in Fig.~\ref{fig:network-topology}. The star topology is characterized by multiple assistants engaging with a sole leader, which is predominantly observed among MPSI protocols such as public-key-based MPSI \cite{bay2021multi,bay2021practical,vos2022fast}, or those utilizing oblivious programmable pseudorandom functions (OPPRF) and oblivious key-value stores (OKVS) \cite{chandran2021efficient,garimella2021oblivious}. Nevertheless, this ``centralized'' configuration results in substantial bandwidth and computational demands on the leader. The mesh topology serves as an enhancement to the star topology by mitigating the communication load and computational pressure on the leader, bearing resemblance to a binary tree structure. This topology is adopted by several of the most efficient MPSI protocols \cite{kolesnikov2017practical,nevo2021simple}. However, the complexity of this network increases, necessitating a robust overall network and sufficient bandwidth for nodes to possibly engage with multiple other nodes. Ring topology, originally proposed by Kavousi et al. \cite{kavousi2021efficient}, involves the sequential transmission of messages among all participants, with both the leader and the assistants assuming comparable protocol roles. In this configuration, the responsibility for the determination of the intersection resides exclusively with the leader at the end of the protocol. Unlike the other two topological frameworks, the ring topology offers distinct advantages: (1) It equitably allocates the communication load among all participants, thus preventing potential bottlenecks associated with a single leader in the star topology. (2) The ring structure reduces the dependency on any single node, thereby improving the resilience of the system to node failures and dropout attacks. (3) It supports scalability for larger networks by minimizing communication overhead, which increases exponentially within the mesh topology as the number of participants increases.

\subsection{Motivation}
In ring topology, embedding multi-party data in a secure end efficient way is quite challenging since each node merely performs two information transmissions, that is, receiving data from the previous node and sending data to the next node. For this reason, each party is required to assemble their own data and received data, then pass it to the next party. Hence, it is critical to carefully ensure that the leader's capability, only determining the overall intersection of all participants' data without revealing any partial intersection.


There are currently two main approaches to build MPSI protocols, partially homomorphic encryption (PHE)-based MPSI and OT-based MPSI. PHE-based protocols typically leverage polynomials or bloom filters to conceal the data sets and then share the coefficients among the participants to calculate the intersection. And PHE is employed to safeguard the security of coefficients throughout the sharing process. However, the coefficient sharing process depends on interactions among multiple participants, inherently forming a mesh topology, thus making the implementation of a ring communication topology be a tough task. On the other hand, most of the existing OT-based MPSI protocols rely on fundamental cryptographic primitives to obfuscate the data sets of the parties, such as oblivious pseudorandom function (OPRF)~\cite{kolesnikov2016efficient}, oblivious programmable PRF (OPPRF)~\cite{kolesnikov2017practical}, and oblivious key-value stores (OKVS)~\cite{chandran2021efficient}. These cryptographic primitives are inherently based on a two-party setup. As a result, to compute the intersection, it is essential to share the intermediate results produced by these primitives among various parties. The secret sharing scheme is leveraged to protect the sharing process. Unfortunately, it is still a protocol within a mesh topology, thereby becoming a bottleneck in realizing the ring topology. Consequently, we raise the following issue: 

\textit{Can we design an efficient MPSI protocol equipped with ring topology?}

\subsection{Contribution}
To cope with the aforementioned issue, in this paper, we design a novel MPSI protocol employing the ring topology, which significantly decreases communication and computational overhead compared to the existing state-of-the-art (SOTA) protocols. Table~\ref{tab:complexity} shows a brief comparison between our MPSI and the SOTA ones on communication and computational overhead. Our contributions are summarized as follows.

\begin{itemize}[itemsep=3pt,topsep=0pt,parsep=0pt,leftmargin=*]
    \item[$\bullet$] We for the first time propose the MPSI protocol with the ring topology alone. Unlike the existing MPSI protocols, Ours maintains identical and independent asymptotic complexities of communication and computation (i.e. $\bigO{N}$ and $N$ denotes the input set size) between the leader and assistant. It effectively reduces the communication and computational burdens on the Leader, which renders him operate at the same level with assistant in regard to the bandwidth and computational capacity.

    \item[$\bullet$] To build the MPSI protocol with the ring topology, we present the notion of multi-point sequential oblivious pseudorandom function (MP-SOPRF). Our MP-SOPRF leverages the output of the previous round of oblivious transfer (OT) for each subsequent node and works in a multi-party setting. Besides, all participants can compute the auxiliary matrix and carry out operations of base OT simultaneously in the initial step. Then in the final OT step, the random OT (ROT) is introduced to substitute the traditional OT for improving the efficiency of protocol.

    \item[$\bullet$] We conducted numerous experiments to validate the effectiveness of our MPSI protocol. Results revealed that our MPSI significantly outperforms the SOTA MPSI protocols in both of communication and computational efficiencies. Concretely, it achieves communication $74.8\%$ reduction and computational efficiency by up to $2.87\times$ in contrast to SOTA ones (15 parties and set size is $2^{20}$).
\end{itemize}
\begin{table}[t]
\centering
\caption{Comparison of similar MPSI protocols, $n$ is the number of parties, $t$ is the number of maximum collusion parties, $N$ is the size of input sets, $h$ is the number of hash functions.}
\label{tab:complexity}
\renewcommand\arraystretch{1.3}
\setlength{\tabcolsep}{1mm}{
\resizebox{\linewidth}{!}{
\begin{tabular}{cccccc}
\hline
\multicolumn{1}{|c|}{\multirow{2}{*}{\textbf{Protocol}}} & \multicolumn{3}{c|}{\textbf{Communication}}         & \multicolumn{2}{c|}{\textbf{Computation}} \\ \cline{2-6}
  \multicolumn{1}{|c|}{}       & \multicolumn{1}{c|}{Topology} & \multicolumn{1}{c|}{Leader} & \multicolumn{1}{c|}{Assistant} & \multicolumn{1}{c|}{Leader}  & \multicolumn{1}{c|}{Assistant} \\ \hline        
\multicolumn{1}{|c|}{KMPRT\cite{kolesnikov2017practical}} & \multicolumn{1}{c|}{Mesh}  & \multicolumn{1}{c|}{\bigO{nN}}  & \multicolumn{1}{c|}{\bigO{tN}} & \multicolumn{1}{c|}{\bigO{n}} & \multicolumn{1}{c|}{\bigO{tN}} \\ \hline
\multicolumn{1}{|c|}{KMS\cite{kavousi2021efficient}}   & \multicolumn{1}{c|}{Ring} & \multicolumn{1}{c|}{\bigO{nN}}  & \multicolumn{1}{c|}{\bigO{hN}} & \multicolumn{1}{c|}{\bigO{nN}}    & \multicolumn{1}{c|}{\bigO{hN}} \\ \hline
\multicolumn{1}{|c|}{NTY\cite{nevo2021simple}}   & \multicolumn{1}{c|}{Mesh}  & \multicolumn{1}{c|}{\bigO{tN}}  & \multicolumn{1}{c|}{\bigO{N}}  & \multicolumn{1}{c|}{\bigO{nN-tN}} & \multicolumn{1}{c|}{\bigO{tN}} \\ \hline
\multicolumn{1}{|c|}{VCE\cite{vos2022fast}}   & \multicolumn{1}{c|}{Star}  & \multicolumn{1}{c|}{\bigO{tN}}  & \multicolumn{1}{c|}{\bigO{N}}  & \multicolumn{1}{c|}{\bigO{nhN}}   & \multicolumn{1}{c|}{\bigO{N}} \\ \hline
\multicolumn{1}{|c|}{Ours}    & \multicolumn{1}{c|}{Ring} & \multicolumn{1}{c|}{\bigO{N}}   & \multicolumn{1}{c|}{\bigO{N}}  & \multicolumn{1}{c|}{\bigO{N}} & \multicolumn{1}{c|}{\bigO{N}} \\ \hline
\end{tabular}}}
\end{table}

\section{Related work}\label{sec:related-work}

Meadows~\cite{meadows1986more} proposed the first PSI protocol. Then a series of lituratures~\cite{kolesnikov2016efficient,pinkas2019spot,chase2020private,rindal2021vole, raghuraman2022blazing,bui2022private} explored PSI with various features and better performance. Freedman et al.~\cite{freedman2005keyword} presented the first MPSI protocol, which is an extension of PSI in the multi-party setting. Since then, numerous efficient MPSI protocols have been proposed. They can be primarily divided into three categories, namely, public key-based MPSI, OT-based MPSI and other MPSI.

Kissner and Song~\cite{kissner2005privacy} presented an MPSI protocol based on the polynomial root representation. In their protocol, After encoding the sets using the roots of a polynomial, each party encrypts the coefficients using an additive homomorphic encryption system, such as Paillier encryption. 
Li and Wu~\cite{li2007unconditionally} employed the secret sharing scheme to replace the homomorphic encryption. 
Subsequently, a series of works~\cite{patra2009information, patra2009round} made efforts to improve the Li and Wu~\cite{li2007unconditionally}'s protocol. 
Ruan et al.~\cite{ruan2019new} put forward the first bitset-based PSI protocol. 
Bay et al.~\cite{bay2021multi, bay2021practical} further extended Ruan's protocol \cite{ruan2019new} to support a multi-party setting, enabling to resist the collusion attack among up to $n-1$ parties.
In 2022, Vos et al.~\cite{vos2022fast} strengthened the 'AND' operation for private set elements based on elliptic curves, making MPSI protocols sutiable for both of the small and large sets.

OT-based MPSI protocols utilize OT as the fundamental tool during the protocol execution. More specific, there are two types for these MPSI protocols according to the role of OT in the entire protocol. Fristly, OT works a underling building block to construct other cryptographic tools such as OPRF, OPPRF, OKVS, etc. Then these tools are involved in the MPSI protocol. Secondly, OT combined with garbled bloom filter (GBF) is utilized to build the MPSI protocol. In 2017, Kolesnikov et al.~\cite{kolesnikov2017practical} brought forward an efficient MPSI protocol and provided an implementation of it. It is the first coded MPSI protocol, which is significant in the development of MPSI protocols. 
Chandran et al.~\cite{chandran2021efficient} strengthened  Kolesnikov et al.'s~\cite{kolesnikov2017practical} protocol, removing the expensive XOR-based secret sharing scheme and then replacing it with the efficient Shamir's secret sharing scheme.
Garimella et al.~\cite{garimella2021oblivious} proposed OKVS  to reduce the communication costs required in OPPRF-based protocols.
Nevo et al.~\cite{nevo2021simple} designed novel MPSI protocols based on OKVS. Their protocols consider the condition that the participants collude after being corrupted by a malicious adversary.
Inbar et al.~\cite{inbar2018efficient} provided three protocols that extend the PSI protocol based GBF\cite{dong2013private} to a multi-party setting.
Recently, Ben-Efraim et al.~\cite{ben2022psimple} further extended Inbar et al's protocol\cite{inbar2018efficient} to achieve security in the malicious model.

Apart from the aforementioned traditional MPSI protocols, there also exist MPSI protocols focusing on achieving special functionalities. Extensive works~\cite{freedman2004efficient, ghosh2019communication, badrinarayanan2021multi, liu2023efficient} explored to threshold MPSI, which adds threshold constraints to traditional multi-party PSI. In the meantime, some other works~\cite{mohassel2020fast, debnath2021secure, fenske2022accountable, trieu2022multiparty, yang2024empsi} studied on MPSI with cardinality (MPSI-CA), supporting multiple parties to collaboratively calculate the size of their set intersection without revealing the actual intersection elements.

\section{Preliminaries}\label{sec:preliminaries}
\subsection{Notation}
We use $\lambda, \sigma$ to denote the computational and statistical security parameters, and use $[n]$ to represent the set of positive integers less than n, i.e., $\{1,2,3,\dots,n\}$. The symbols $P_1,\dots,P_n$ denote the set of $n$ parties, where each party owns the input set $X_i$ for $i\in [n]$. $P_1$ is assumed to be the leader who calculates the final intersection results, while other parties are regarded as assistants. For a vector $v$, we use $v[j]$ to denote its $j$-th element. For $n\times m$ size matrix $M$, we use $M_j$ to denote its $j$-th column where $j\in [m]$. $M^i$ represents the matrix $M$ belongs to party $P_i$ and $||x||_\mathrm{H}$ denotes the hamming weight of string $x$. The symbol $\negl[\lambda]$ denotes a negligible function that $\mathsf{negl}(\lambda)<1/p(\lambda)$ holds for any polynomial $p(\cdot)$.

\subsection{Random oblivious transfer}
Oblivious Transfer (OT)~\cite{rabin2005exchange, naor2001efficient} is a fundamental cryptographic protocol allowing the sender to transfer one of multiple messages to the receiver without revealing which message was sent, and the receiver learns nothing about the other messages. Fig.~\ref{fig:FOT} depicts the functionality of 1-out-of-2 OT, where the sender inputs two messages $(m_0,m_1)$ and the receiver inputs a chosen bit $c$. As a result, the receiver learns $m_c$ without knowing any information about $m_{1-c}$ and the sender learns nothing about $c$. In our MPSI protocol, we introduce a variant of 1-out-of-2 OT, i.e., random OT (ROT)~\cite{beaver1992efficient}, where the sender inputs noting and the receiver inputs a chosen bit $c$, then the sender gets two random string $(m_0,m_1)$ and the receiver gets $m_c$, thus reducing the communication overhead. Fig.~\ref{fig:FROT} shows the functionality of ROT.


\begin{figure}[htb]
    \centering
    \fbox{\begin{minipage}{\linewidth}
            \textbf{Parameters:} Message length $\ell$.

            \textbf{Input:} The sender inputs two message $(m_0,m_1)$ and the receiver inputs a choice bit $c\in \bin$.

            \textbf{Output:} the sender gets nothing and  the receiver gets $m_c$.
        \end{minipage}
    }
    \caption{Ideal functionality for Oblivious Transfer $\mathcal{F}_\mathrm{OT}$ }
    \label{fig:FOT}
\end{figure}

\begin{figure}[htb]
    \centering
    \fbox{\begin{minipage}{\linewidth}
            \textbf{Parameters:} Message length $\ell$.

            \textbf{Input:} The sender inputs nothing and the receiver inputs a choice bit $c\in \bin$.

            \textbf{Output:} the sender gets message $(m_0,m_1)$ where $m_i \in \bin ^\ell$ and the receiver gets $m_c$.
        \end{minipage}
    }
    \caption{Ideal functionality for Random Oblivious Transfer $\mathcal{F}_\mathrm{ROT}$ }
    \label{fig:FROT}
\end{figure}

\subsubsection{Hamming correlation robustness}
The security of our MPSI protocol relies on the \textit{correlation robustness assumption}~\cite{pinkas2019spot,chase2020private}, we extend this assumption to a new definition which we call $t\text{-}party$ correlation robustness assumption to support the security of our protocol.

\begin{definition}[Hamming Correlation Robustness]
    \label{def:hamming}
    A hash function $H$ with input length $n$ is $d\text{-}$hamming correlation robust if for any $a_1,\dots,a_m, b_1,\dots,b_m \in \bin^n$ with $||b_j||_H \geq d$ where $j\in [m]$, the following distribution induced by $s\overset{\$}{\leftarrow} \bin ^n$ is pseudorandom. Namely, for a random function $F$, we have $(H(a_1\xor[b_1\cdot s]), \dots, H(a_m\xor[b_m\cdot s])) \cindist(F(a_1\xor[b_1\cdot s]), \dots, F(a_m\xor[b_m\cdot s]))$, where $\cdot$ denotes bit-wise AND and $\xor$ denotes bit-wise XOR.   
\end{definition}

\begin{corollary}[t-party Hamming Correlation Robustness]\label{col:t-party-hamming-correlation}
    A hash function $H$ with input length $n$ and $t\geq 2$ parties is $t\text{-}party$ $d\text{-}$hamming correlation robust if for any $a_1,\dots,a_m, b^i_1,\dots,b^i_m \in \bin^n$ with $||b^i_j||_H \geq d$ where $j\in [m]$ and $i\in[t-1]$, the following distribution induced by random sampling of $s_i\overset{\$}{\leftarrow} \bin ^n$ where $i\in[t-1]$ is pseudorandom. Namely, for a random function $F$, we have
    \begin{align*}
    \vspace{-10pt}
        (&H(a_1\xor[b^1_1\cdot s_1]\xor[b^2_1\cdot s_2]\dots\xor[b^{t-1}_1\cdot s_{t-1}]),\dots,\\
        &H(a_m\xor[b^1_m\cdot s_1]\xor[b^2_m\cdot s_2]\dots\xor[b^{t-1}_m\cdot s_{t-1}])) \\
        \cindist (&F(a_1\xor[b^1_1\cdot s_1]\xor[b^2_1\cdot s_2]\dots\xor[b^{t-1}_1\cdot s_{t-1}]),\dots,\\
        &F(a_m\xor[b^1_m\cdot s_1]\xor[b^2_m\cdot s_2]\dots\xor[b^{t-1}_m\cdot s_{t-1}])),
        \vspace{-10pt}
    \end{align*}
    where $\cdot$ denotes bit-wise AND and $\xor$ denotes bit-wise XOR.
\end{corollary}


\begin{proof}
    Let $a_1,\dots,a_m, b^i_1,\dots,b^i_m \in \bin^n$ where $\forall j\in [m]$ and $\forall i\in[t-1]$, $||b^i_j||_H \geq d$ holds, set $s_i \overset{\$}{\leftarrow} \bin ^n$.
    For $j\in [m]$, calculate $c_j = a_j\xor[b^1_j\cdot s_1]\xor[b^2_j\cdot s_2]\dots\xor[b^{t-2}_j\cdot s_{t-2}]$, then we have $c_1,\dots,c_m, b^{t-1}_1,\dots,b^{t-1}_m\in \bin ^n$ with $||b^{t-1}_j||_\mathrm{H} \geq d$ for $j\in[m]$ and $s_{t-1} \overset{\$}{\leftarrow} \bin ^n$. This satisfies the selection range of each parameter in Definition~\ref{def:hamming}. Therefore, through Definition~\ref{def:hamming}, we have
    \begin{align*}
    \vspace{-10pt}
        &(H(c_1\xor[b^{t-1}_1\cdot s_{t-1}]), \dots, H(c_m\xor[b^{t-1}_m\cdot s_{t-1}])) \\
        \cindist \ &(F(c_1\xor[b^{t-1}_1\cdot s_{t-1}]), \dots, F(c_m\xor[b^{t-1}_m\cdot s_{t-1}])).
    \vspace{-10pt}
    \end{align*}

    By expanding $c_j~(j\in [m])$ we can yield the equivalence relationships that follow
    \begin{align*}
        (&H(a_1\xor[b^1_1\cdot s_1]\xor[b^2_1\cdot s_2]\dots\xor[b^{t-1}_1\cdot s_{t-1}]), \dots, \\
        &H(a_m\xor[b^1_m\cdot s_1]\xor[b^2_m\cdot s_2]\dots\xor[b^{t-1}_m\cdot s_{t-1}])) \\
        \cindist (&F(a_1\xor[b^1_1\cdot s_1]\xor[b^2_1\cdot s_2]\dots\xor[b^{t-1}_1\cdot s_{t-1}]), \dots, \\
        &F(a_m\xor[b^1_m\cdot s_1]\xor[b^2_m\cdot s_2]\dots\xor[b^{t-1}_m\cdot s_{t-1}])),
    \end{align*}
    and this concludes the proof.
\end{proof}

\subsection{Multi-Point OPRF}
Oblivious Pseudorandom Function (OPRF) combines the concepts of OT and pseudorandom function (PRF), in OPRF, the sender inputs nothing and the receiver inputs $y_1,\dots,y_n \in Y$, as a result, the sender obtains a PRF key $k$ and the receiver obtains the PRF values $\mathsf{OPRF}_k(y_1), \dots, \mathsf{OPRF}_k(y_n)$ about its input. The sender learns nothing about the receiver's input $y_1,\dots,y_n$ and the receiver learns nothing about the PRF key $k$. Based on OPRF, multi-point OPRF (MP-OPRF)~\cite{chase2020private} is achieved through the following approaches: The sender with the input set $X$ picks a random bit string $s \overset{\$}{\leftarrow} \bin^w$ of length $w$ and the receiver with the input set $Y$ generates two binary $m\times w$ matrices $A$ and $B$. $A$ is a random matrix, and $B = A \oplus D$ where $D$ embeds the information of the input set $Y$. After running $w$ OTs between parties where the sender acts as the OT receiver, the sender obtains a $m\times w$ matrix $C$ with each column of $C$ being $A_i$ or $B_i$ depending on the chosen bit $s_i$. The PRF key is the matrix $C$ and the OPRF value $\psi$ is calculated as $\psi = H(C_1[v[1]]||\dots||C_w [v[w]])$ where $v = F_k(x)$, then the sender sends $\psi$ it to the receiver. The receiver evaluates the OPRF value as $\phi = H(A_1[v[1]]||\dots||A_w [v[w]])$ where $v = F_k(y)$. Consequently, $\psi = \phi$ indicates $x=y$ as $A_i[v[i]] = B_i[v[i]] = C_i[v[i]]$. Otherwise, from receiver's perspective, the value of $\psi$ appears pseudorandom.

\subsection{Security model}
We adopt a semi-honest security model to assess the security of our proposed MPSI protocol. Fig.~\ref{fig:FMPSI} shows the ideal functionality of our MPSI.

\begin{figure}[htb]
\vspace{-5pt}
    \centering
   \fbox{\begin{minipage}{0.95\linewidth}
\textbf{Parameters:} Party number n and the upper bound of party input set size $N$.

\textbf{Input:} Each party $P_i$ has an input set
$$X_i = \{ x^i_1,\dots,x^i_{n_i} \},$$
where $x^i_j \in \bin ^*$ for $j\in[n_i]$.

\textbf{Output:} Party $P_1$ obtain the intersection $$I = X_1 \cap \dots \cap X_n.$$
\end{minipage}}
\caption{Ideal functionality for MPSI $\mathcal{F}_\mathrm{MPSI}$ }
\label{fig:FMPSI}
\end{figure}

\subsubsection{Adversary model}
In the semi-honest security model, the adversary will adhere to all protocol specifications while attempting to maximize its understanding of the information it observes. The adversary's goal is to gather details regarding other participants' local data or the outcomes of the protocol. It is important to note that the leader is not colluding with any assistants, a standard condition in existing literature \cite{zhang2019efficient,abadi2017efficient}. Under this assumption, our protocol demonstrates resilience against corruption up to $n-1$.

\begin{definition}[Semi-Honest Security]
    Let the view of $P_i$ in protocol $\Pi$ be as $view_i^\Pi (X_1, \dots, X_n)$. Let $\mathcal{F}(X_1, \dots, X_n)$ be the output of $P_1$ in ideal functionality. Let $out^\Pi(X_1, \dots, X_n)$ be the output of $P_1$ in the protocol. The protocol $\Pi$ is semi-honest secure if there exist PPT simulators $\mathcal{S}_1$ and $\mathcal{S}_2$ such that for all inputs $X_1, \dots, X_n$, $\{\mathcal{S}_2(\lambda,X_i),(X_1,\cdots,X_n)\} \overset{c}{\approx} \{ view_i^\Pi (X_1, \dots, X_n),  out^\Pi(X_1, \dots, X_n)\}$, and for $i\in [2,n]$, $\mathcal{S}_1(\lambda, X_1, (X_1, \cdots, X_n)) \overset{c}{\approx} view_1^\Pi (X_1, \dots, X_n)$.
\end{definition}

\section{Construction}\label{sec:construction}
\subsection{Technical overview} 
To construct a novel MPSI protocol with ring topology, we first extend the MP-OPRF protocol to the MP-SOPRF protocol to fit the multi-party setting. As MP-OPRF relies on the basic 1-out-of-2 OT, and in multi-party setting, $P_i$ have to wait $P_1,\dots,P_{i-2}$ to complete before it can interact with $P_{i-1}$. Thus, to minimize the waiting time in the protocol execution, we replace traditional OT with ROT. Through ROT, both parties initiate a set of random bits $r_0, r_1$ based on a random bit $b$, where the sender gets $r_0, r_1$ and the receiver gets $r_b$, and no further interaction is needed. After that, the sender masks its own messages $m_0,m_1$ with $r_0, r_1$ as $r_0\xor m_0, r_1\xor m_1$ and sends the masked values to the receiver. Since the receiver only has $r_b$, it can only reveal one of the two messages $m_b = r_b \xor c_b$. As the initialization process is independent of the messages, all participants are able to execute this process in advance (or in parallel). We use two forms of ROT here. The first form is conducted between $P_1$ and $P_2$, where P1 is the leader. In the subsequent process, $P_1$ chooses a random message $m_0$ as the input of OT. In our approach, $P_1$ directly utilizes the output $r_0$ of ROT as $m_0$, making the communication overhead halved at the beginning of the protocol ($P_1$ only needs to send $c_1$ to $P_2$). On the other hand, $P_2, \dots , P_n$ just execute the normal ROT protocol.

\subsection{Construction of MP-SOPRF}
We regard multi-point sequential oblivious pseudorandom function (MP-SOPRF) as a special MP-OPRF protocol for multi-party scenarios. Let $X_i$ denote the data set of party $P_i$ where $i\in [n-1]$. The input of MP-SOPRF is no longer a single party's data set, but multiple parties' data sets, i.e., $\{X_1,\dots, X_{n-1}\}$. After MP-SOPRF, $P_1$ receives the OPRF values corresponding to the data in $X_1$ and $P_n$ receives the OPRF key. Let $I = \{X_1\cap\dots\cap X_{n-1}\}$, each value obtained by $P_1$ is meaningful only when $x \in I $, otherwise it is a random string. Note that $P_1$ cannot distinguish whether the obtained values are meaningful values or random values. Assuming matrix $C$ is a PRF key. The pseudorandom function can be computed as $\mathsf{OPRF}_{C}(x) = H(C_1[v[1]]||\dots||C_w[v[w]])$. Finally, $P_1$ obtains a matrix $A$ and computes the output values by $\mathsf{OPRF}_A(x)$. On the other hand, $P_n$ obtains the PRF key $C$. For $x \in I$, we have $ \mathsf{OPRF}_A(x) = \mathsf{OPRF}_C(x)$, otherwise if $x \in X_1\setminus I$, $\mathsf{OPRF}_C(x)$ should be a random string in $P_1$'s view. The functionality of MP-SOPRF is defined in Fig.~\ref{fig:mp-soprf}.

\begin{figure}[htb]
    \centering
   \fbox{\begin{minipage}{0.95\linewidth}

\textbf{Input:} Each party $P_i$ has an input set $X_i$ for $i\in[n-1]$.

\textbf{Output:} Party $P_n$ obtain the MP-OPRF key $k$. Party $P_1$ obtain the MP-OPRF value set $V(x)$ for each $x$ in his input set $X_1$. If and only if $x\in \{X_1\cap\dots\cap X_{n-1}\}$, there is $\mathsf{OPRF}_k(x) = V(x)$
\end{minipage}}
\caption{Ideal functionality for MP-SOPRF $\mathcal{F}_\mathsf{MP\text{-}SOPRF}$}
\label{fig:mp-soprf}
\end{figure}

In MP-SOPRF, the intersection determination method is reusable, and regardless of how the random seed of the OT receiver is chosen, if the data is in the intersection, the value at position $M_0$ of the OT sender's inputs ($M_0$ and $M_1$) will be obtained, because when designing the OT input matrices $M_0$ and $M_1$, the values at corresponding positions in $M_0$ and $M_1$ for data mapped from the OT sender are the same. This ensures that in the OT results obtained by the OT receiver, as long as the data used belongs to the OT sender, the computed results will match the values in $M_0$ at the corresponding positions. 

Thus, when identifying intersections, the OT sender (who acts as the PSI receiver in PSI protocols) only needs to compute and judge through its self-selected random number. This approach is also feasible in a multi-party structure. Only when the data is in the intersection of all parties, regardless of how any party's OT random seed is selected,  the outcome of the entire OT link will be $m_0$. Finally, the leader (also referred to as the receiver in two-party PSI) only needs to match its generated random matrix $A$ with the matrix $C$ sent by $P_n$ during the final intersection matching. When the leader's data $x$ is in the intersection of all parties will occur $A(x) = C(x)$.

\begin{figure*}[htbp]
\centering
    \includegraphics[width=\linewidth]{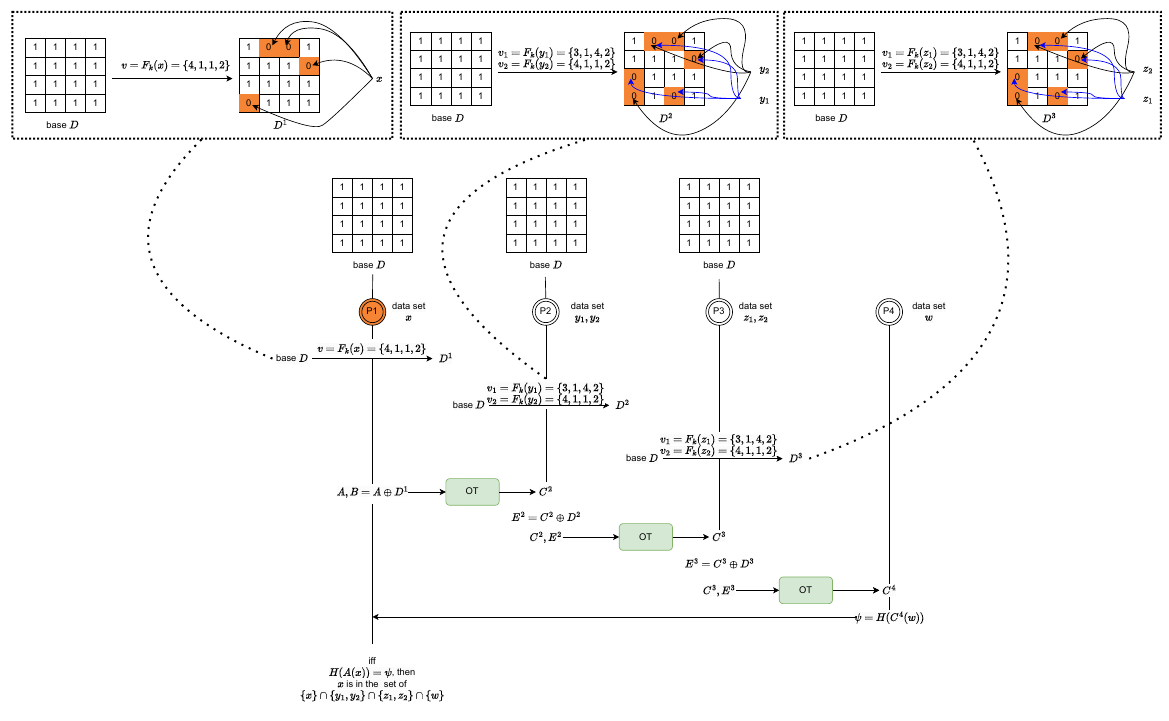}
    \caption{Overview of our MPSI protocol}
    \label{fig:framework}
\end{figure*}

\subsection{Construction of our MPSI protocol}
We first give a highlight overview of our MPSI protocol consisting of 4 parties, as is shown in Fig.~\ref{fig:framework}, $P_1$ is the leader with data set $x$. $P_2,P_3,P_4$ are assistants and the data sets are $\{ y_1, y_2\}, \{ z_1, z_2\},\{ w\} $.  $P_1$ prepare an all-ones binary matrix $D$ and compute a vector $v=F_k(x)$. Subsequently, based on the vector $v$, the corresponding elements in $D$ are set to 0 and the result is denoted as $D^1$. $P_2$ and $P_3$ generate matrices $D^2$ and $D^3$ using the method described above based on their own datasets. It should be noted that if a corresponding position in the matrix is already 0, no further action is needed. Firstly, $P_1$ randomly generates a binary matrix $A$ and runs OT with $P_2$, where the OT input is $A, B = A \oplus D^1$. Let $C^2$ be the OT result obtained by $P_2$ and then $P_2$ runs OT with $P_3$, where the OT input is $C^2, E^2 = C^2 \oplus D^2$ and the OT output is $C^3$ . Following this approach, $P3$ and $P_4$ continue to run OT, with the OT input being $C^3, E^3 = C^3 \oplus D^3$ and the OT output being $C^4$. Afteward, $P_4$ computes the vector $\psi$ using $C^4$ and its own data set $\{w\}$, and sends it to $P_1$. Finally $P_1$ computes $H(A(x))$ using $A$  and its own data set $\{x\}$ and outputs $x$ is in the intersection if $H(A(x)) = \psi$. Regardless of how the random seeds for the OT are selected by the participating parties, there are the following relationships between the matrices. $A(x) = C^2(y)$ if and only if $x = y$. $C^2(y) = C^3(z)$ if and only if $y = z$. $C^3(y) = C^4(w)$ if and only if $z = w$. Therefore, $A(x) = C^4(w)$ if and only if $x = y = z = w$ which means that $x$ is in the set of $\{ x\} \cap \{ y_1, y_2\}\cap\{ z_1, z_2\}\cap\{ w\}$.

Generally, a multiparty PSI protocol involves $n$ participants $P_1,\dots ,P_n$, each having a dataset set $X_i$. They aim to compute the intersection of theses datasets $X_1\cap \dots \cap X_n$ without revealing any information beyond the common intersection among the parties. We consider $P_1$ as the leader party that gets the intersected set. we describe the flow of our protocol in Figure~\ref{fig:construction}. In Step \ref{OT1} and \ref{OT2}, involving the execution of OT protocols between two parties, to enhance efficiency, we can instantiate it using random OT. Specific procedures are described in Figures \ref{fig:ROT1} and \ref{fig:ROT2}, respectively. 

\begin{figure*}[htb]
\centering
\fbox{
\begin{minipage}{0.98\linewidth}
\textbf{Parameters:}  Parties $P_1$, $P_2$,$\dots$, and $P_n$ agree on security parameters $\lambda$ and $\sigma$, two hash functions $H_1 : \{\mathrm{0,1}\}^* \rightarrow \{\mathrm{0,1}\}^{\ell_1}$ and $H_2 : \{\mathrm{0,1}\}^\omega \rightarrow \{\mathrm{0,1}\}^{\ell_2}$, pseudorandom function $F : \{\mathrm{0,1}\}^\lambda \times \{\mathrm{0,1}\}^{\ell_1} \rightarrow [m]^\omega$ where $\ell_1$, $\ell_2$, $m$, and $\omega$ are protocol parameters common to all parties.

\textbf{Preprocessing Stage:}
\begin{enumerate}
    \item Each party $P_i$ samples a random string $ s_i \sample \bin ^\omega$ as the OT receiver inputs for $i\in \{2,3,\dots,n\}$.
    \item All participating parties share the PRF key $k \sample \bin ^\lambda$.
    \item Each party $P_i$ generates an $m\times \omega$ matrix $D^i$ with all elements equal to 1 for all $i\in [n]$.
    \item For each $x\in X_i$, party $P_i$ computes $v = F_k(H_1(x))$ and set $D_j^i[v[j]] = 0$ for all $j\in[\omega]$ where $i\in [n-1]$. 
\end{enumerate}

\textbf{Wheeled Oblivious Transfer Interaction:}
\begin{enumerate}[resume]
    \item $P_1$ samples a random binary matrix $ A \overset{\$}{\leftarrow} \bin ^{m\times\omega}$. Computes matrix $B = A\oplus D^1$.
    \item\label{OT1}  $P_1$ as the OT sender with input $\{A_j, B_j\}_{j\in\omega}$ and $P_2$ as the OT receiver with input $s_2$ run OT $\omega$ times. After all OTs are completed, $P_2$ obtains an ${m\times\omega}$ matrix $C^2$ whose column vectors is a $m$-bit OT result string.
    \item $P_2$ computes $E^2 = C^2 \oplus D^2$.
    \item\label{OT2} Next each party $P_i$ runs wheeled OT interaction  for $i\in\{2,3,\dots,n-1\}$, which means that party $P_i$ as the OT sender runs OT with the next party $P_{i+1}$ as described in step 6. The OT result matrix for party $P_{i+1}$ is denoted as $C^{i+1}$. When party $P_i$ is the OT sender, his input is $\{C^i, E^i = C^i \oplus D^i \}$ where matrix $C^i$ is the outcome of OT conducted between $P_i$ as the OT receiver and $P_{i-1}$. 

\end{enumerate}

\textbf{Intersection Result Calculation:}
\begin{enumerate}[resume]
    \item  \label{send psi} After Wheeled OT, party $P_n$ will get a matrix $C^n$, then $P_n$ computes $v = F_k(H_1(x))$ for each $x\in X_n$ and its OPRF value $\psi = H_2(C^n_1[v[1]]||\dots||C^n_\omega [v[\omega]])$ and sends $\psi$ to $P_1$.
    \item \label{h2} The set of elements received by $P_1$ from $P_n$ is denoted as $\Psi$. $P_1$ computes $v = F_k(H_1(x))$ for each $x\in X_1$ and his own OPRF value as $\phi = H_2(A_1[v[1]]||\dots||A_\omega [v[\omega]])$. If $\phi \in \Psi$, $P_1$ puts $x$ into the intersection set $I$.
    \item Party $P_1$ output $I$ as the intersection of all parties. 
\end{enumerate}
 \end{minipage}}
    \caption{Our multiparty private set intersection protocol}
    \label{fig:construction}
\end{figure*}

\begin{figure}
    \centering
\fbox{
\begin{minipage}{0.95\linewidth}
    \begin{enumerate}
    \item Party $P_1$ and $P_2$ perform $\omega$ instances random OTs  with message length $m$ where $P_2$ is the OT receiver with input bits $s_2\in \bin^\omega$. At the end of random OT, $P_1$ obtains $\omega$ message pairs $\{r_j^0, r_j^1\}_{j\in[\omega]}$. $P_2$ obtains $\omega$ random message $\{r_j\}_{j\in[\omega]}$ where $r_j = r_j^{s_2[j]}$.
    \item $P_1$ constructs matrix $A$, set the column vectors of $A$ as $\{r_j^0\}_{j\in[\omega]}$, then compute matrix $B = A \oplus D^1$. Then it calculates matrix $\Delta_j = B_j \oplus r_j^1$ for all $j\in[\omega]$ and send $\Delta = \Delta_1 || \dots || \Delta_\omega$ to $P_2$.
    \item $P_2$ constructs matrix $C^2$ as:  
    \begin{align*} 
        \mathrm{for} \  j\in [\omega],\ C^2_j =  \left\{\begin{array}{lr}  
             r_j, \quad\quad\quad \ s_2[j] = 0,\\
             r_j \oplus \Delta_j, \quad \mathrm{otherwise.}
    \end{array} \right.
    \end{align*}
    \end{enumerate}
\end{minipage}
}
    \caption{The instantiated using random OT in Step \ref{OT1}}
    \label{fig:ROT1}
\end{figure}

\begin{figure}
    \centering
    \fbox{
    \begin{minipage}{0.95\linewidth}
    \begin{enumerate}
    \item $P_i$ and $P_{i+1}$ perform $\omega$ instances ROTs with message length $m$ where $P_{i+1}$ is the OT receiver with input bits $s_{i+1}\in \bin^\omega$. At the end of random OT, $P_i$ obtains $\omega$ message pairs $\{r_j^0, r_j^1\}_{j\in[\omega]}$. $P_{i+1}$ obtains $\omega$ random message $\{r_j\}_{j\in[\omega]}$ where $r_j = r_j^{s_{i+1}[j]}$.
    \item $P_i$ does following steps:
    \begin{enumerate}
        \item Construct matrix $R_0$ which column vectors are $\{r_j^0\}_{j\in[\omega]}$ and matrix $R_1$ which column vectors are $\{r_j^1\}_{j\in[\omega]}$. Then compute matrix $E^i = C^i \oplus D^i$.
        \item Compute matrix $\Gamma = R_0 \oplus C^i$, $\Delta = R_1 \oplus E^i$ and send $\Delta$, $\Gamma$ to $P_{i+1}$.
    \end{enumerate}
    \item $P_{i+1}$ constructs matrix $C^{i+1}$ as:
    \begin{align*}
   \text{for}\ j \in [\omega], C^{i+1}_j =  \left\{\begin{array}{cc}  
             r_j \oplus \Gamma_j   &  \mathrm{if} \quad s_{i+1}[j] = 0\\
             r_j \oplus \Delta_j & \mathrm{otherwise}       
    \end{array} \right.
    \end{align*}
 \end{enumerate}
 \end{minipage}
}
    \caption{The instantiated using random OT in Step \ref{OT2}}
    \label{fig:ROT2}
\end{figure}

\subsection{Correctness}
At a high level, for each $x\in X_1$, let $v = F_k(H_1(x)$, we can get $A_j[v[j]] = B_j[v[j]]$ for all $j\in[\omega]$. After OT between $P_1$ and $P_2$, regardless of what $P_2$ choose for $s_2$, $A_j[v[j]] = B_j[v[j]] = C^2_j[v[j]]$ for all $j\in[\omega]$ iff $x\in X_1 \cap X_2$. 
Through a similar analysis, we can conclude that after OT between $P_i$ and $P_{i+1}$, let $v = F_k(H_1(x)$ regardless of what $P_{i+1}$ choose for $s_{i+1}$, $C^i_j[v[j]] = E^i_j[v[j]] = C^{i+1}_j[v[j]]$ for all $j\in[\omega]$ iff $x\in X_{i} \cap X_{i + 1}$ for all $i\in \{2,3,\dots,n-1\}$. At the end of the entire protocol flow, $P_n$ will get matrix $C^n$. By transitivity, let $v = F_k(H_1(x)$ we can deduce $C^n_j[v[j]] = E^{n-1}_j[v[j]] = C^{n-1}_j[v[j]] = \dots = E^{2}_j[v[j]] = C^{2}_j[v[j]] = B_j[v[j]] = A_j[v[j]]$ for all $j\in[\omega]$ iff $x \in X_1\cap \dots \cap X_n$, which means that $C^n_j[v[j]] = A_j[v[j]]$. Then for each $x\in X_1$, We can conclude $x\in X_1\cap \dots \cap X_n$ iff $\phi \in \Psi$.

\section{Security analysis}
As we assume the leader will not collude with any assistant, we consider the security of our MPSI protocol under two scenarios: 1) only the leader is corrupted, and 2) the assistants are corrupted (i.e., the leader is not corrupted). Our security analysis relies on the following two lemmas~\cite{chase2020private}.
\begin{lemma}
    If $F$ is a pseudorandom function, $H_1(x)$ is different for each $x \in X_1\cup \dots \cup X_n$, the probability that the number of ``1''s in the sequence $D^i_1[v[1]],\dots,D^i_w[v[w]]$ does not exceed $d$ (the parameter of Hamming correlation robustness defined in Corollary~\ref{col:t-party-hamming-correlation}) is negligible.
    \label{lem:1}
\end{lemma}
\begin{lemma}
    If $H_2$ satisfies $t\text{-}$party Hamming correlation robustness, for $x\in X_1,  y \in X_n \backslash I$, $x\neq y$, $v=F_k(H_1(x))$, $u=F_k(H_1(y))$, the probability that $H_2(A_1[v[1]]||\dots||A_w[v[w]]) = H_2(C^n_1[u[1]]||\dots||C^n_w[u[w])$ holds is negligible.
    \label{lem:2}
\end{lemma}

\begin{theorem}
    If $F$ is a pseudorandom function, $H_1$ is a collision resistant hash function, and $H_2$ is a $t\text{-}$party $d\text{-}$Hamming correlation robust hash function, then our MPSI protocol securely realizes $\mathcal{F}_{\mathrm{MPSI}}$ (Figure \ref{fig:FMPSI}) in the semi-honest mode. 
\end{theorem}

\begin{proof}[$P_1$ is not corrupted]
We construct a probabilistic polynomial time (PPT) simulator $\mathcal{S}_Z$ where $Z$ is a subset of which the parties in $P_2,\dots,P_n$ are corrupted by the adversary. Given $\{ X_i \}_{i\in Z}$, $\mathcal{S}_Z$ can generate simulated views, and these simulated views are computationally indistinguishable from the views of the joint distribution of corrupted parties in the real protocol. $\mathcal{S}_Z$ generates the random string $\{s_i \overset{\$}{\leftarrow} \bin^w\}_{i\in Z}$ honestly and chooses the random matrix $\{C^i\}_{i\in Z} \in \bin ^{m\times w}$. It runs the OT simulator to simulate the view of the OT receiver for the corrupted party $P_i\in Z$ with inputs $s_i[1],\dots,s_i[w]$ and outputs $C_1^i,\dots,C_w^i$. Furthermore, $\mathcal{S}_Z$ sends a uniform random PRF key $k$ to the corrupted parties. Finally $\mathcal{S}_Z$ outputs the view of corrupted parties. We prove that $\{ \sdv_Z(\lambda , \{X_i\}_{i\in Z}), \mathcal{F}(X_1, \dots, X_n)\}  \cindist \{ view_Z^\Pi (X_1, \dots, X_n), out^\Pi(X_1,\dots,X_n) \}$ via a sequence of computationally indistinguishable hybrid arguments.

\noindent$\mathbf{Hybrid_0:}$ The view of corrupted parties $\{P_i\}_{i\in Z}$ and the output of $P_1$ in the real protocol.

\noindent$\mathbf{Hybrid_1:}$ Same as $\mathbf{Hybrid_0}$ except that for party $P_2$, for each $j\in[w]$, if $s_2[j] = 0$, $P_1$ samples a random column $A_j\overset{\$}{\leftarrow}\bin^m$ and computes $B_j = A_j \oplus D^1_j$; otherwise it samples a random column $B_j\overset{\$}{\leftarrow}\bin^m$ and computes $A_j = B_j \oplus D^1_j$. This hybrid is identical to $\mathbf{Hybrid_0}$.

\noindent$\mathbf{Hybrid_2:}$ Same as $\mathbf{Hybrid_1}$ except that $\sdv_Z$ chooses a random PRF key $k$. This is statistically identical to $\mathbf{Hybrid_1}$.

\noindent$\mathbf{Hybrid_3:}$ Same as $\mathbf{Hybrid_2}$ except that the protocol aborts if there exist $x\neq y\in X_1\cup\dots\cup X_n$ such that $H_1(x) = H_1(y)$. This hybrid is identical to $\mathbf{Hybrid_2}$ except negligible probability because $H_1$ is collision resistant and then the aborting probability is negligible.

\noindent$\mathbf{Hybrid_4:}$ Same as $\mathbf{Hybrid_3}$ except that the protocol aborts if there exist $x\in X_n \backslash I$, for $v = F_k(H_1(x))$, there are fewer than $d$ 1's in $D^i_1[v[1]],\dots,D^i_w[v[w]]$ where $i\in[n-1]$. The choice of parameter $m,w$ ensure that if $F$ is a random function and $H_1(x)$ is different for each $x \in X_1\cup \dots \cup X_n$, the probability of abort is negligible by Lemma \ref{lem:1}. 

\noindent$\mathbf{Hybrid_5:}$ Same as $\mathbf{Hybrid_4}$ except that the output of $P_1$ is replaced by $\mathcal{F}(X_1,\dots,X_n) = I = X_1\cap\dots\cap X_n$. This hybrid changes the output of $P_1$ iff there are $x\in X_1,  y \in X_n \setminus I$, $x\neq y$, such that $v=F_k(H_1(x))$, $u=F_k(H_1(y))$, $H_2(A_1[v[1]]||\dots||A_w[v[w]]) = H_2(C^n_1[u[1]]||\dots||C^n_w[u[w]).$ Since $H_2(C^n_1[u[1]]||\dots||C^n_w[u[w])$ is pseudorandom by the $t$ party $d$-Hamming correlation robustness of $H_2$, this hybrid is identical to $\mathbf{Hybrid_4}$ except for a negligible probability for sufficiently large $\ell_2$ by Lemma \ref{lem:2}.

\noindent$\mathbf{Hybrid_6:}$ Same as $\mathbf{Hybrid_5}$ except that the protocol does not abort. From the above discussion, it is evident that the probability of a protocol abort is negligible.

\noindent$\mathbf{Hybrid_7:}$ Same as $\mathbf{Hybrid_6}$ except that for corrupted parties $\{ P_i\}_{i\in Z}$, $\sdv_Z$ samples the matrixes $\{C^i\}_{i\in Z}$ and runs the OT simulator to simulate the view of an OT receiver for $P_i$. This hybrid is computationally indistinguishable from $\mathbf{Hybrid_6}$ by the security of OT protocol.   
\end{proof}

\begin{proof}[$P_1$ is corrupted]
We construct a PPT simulator $\mathcal{S}_1$ as follows. $\mathcal{S}_1$ generate simulated views wiht $P_1$'s input set $X_1$, the max input set size $N$ (Here we can also input the size $n_n$ of $P_n$'s input set, as in traditional MPSI, the size of each party's sets are not information requiring protection and do not affect the correctness of our proof or the security of the protocol) and the intersection $I = \mathcal{F}(X_1,\dots, X_n )$, and these simulated views are computationally indistinguishable from the views of $P_1$ in the real execution of the protocol. $\mathcal{S}_1$ computes the matrix $A$ and $B$ honestly and run the OT simulator to simulate the view of the OT sender for the corrupted party $P_1$. For each $x\in I $, it computes $v = F_k(H_1(x))$ and the OPRF value $\psi = H_2(A_1[v[1]]||\dots||A_w[v[w]])$. Let the set of OPRF values computed by $x\in I $ be $\mathrm{\Psi}_I$. Subsequently it selects $N-|I|$ random 
$\ell_2$-bit random strings and let the set of these strings be $\mathrm{\Psi}_R$. In Step \ref{send psi}, it send $\mathrm{\Psi}_I\cup \mathrm{\Psi}_R$ to $P_1$ and finally outputs the $P_1$' simulated view.
We prove $\sdv_1(\lambda , X_1, \mathcal{F}(X_1, \dots, X_n)) \cindist view_1^\Pi (X_1, \dots, X_n) $ via a sequence of computationally indistinguishable hybrid argument.

\noindent$\mathbf{Hybrid_0:}$ The view of $P_1$ in the real protocol.

\noindent$\mathbf{Hybrid_1:}$ Same as $\mathbf{Hybrid_0}$ except that the protocol aborts if there exist $x\neq y\in X_1\cup\dots\cup X_n$ such that $H_1(x) = H_1(y)$. This hybrid is identical to $\mathbf{Hybrid_0}$ except negligible probability because $H_1$ is collision resistant and then the aborting probability is negligible.

\noindent$\mathbf{Hybrid_2:}$ Same as $\mathbf{Hybrid_1}$ except that the protocol aborts if there exist $x\in X_n \setminus I$ such that for $v = F_k(H_1(x))$, there are fewer than $\lambda$ 1's in $D^1_1[v[1]],\dots,D^1_w[v[w]]$. The choice of parameter $m,w$ ensure that if $F$ is a random function and $H_1(x)$ is different for each $x \in X_1\cup \dots \cup X_n$, the probability of abort is negligible by Lemma \ref{lem:1}.

\noindent$\mathbf{Hybrid_3:}$ Same as $\mathbf{Hybrid_2}$ except that $\sdv_1$ run the OT simulator to simulate the view of an OT sender for $P_1$. This hybrid is computationally indistinguishable from $\mathbf{Hybrid_2}$ by the security of OT protocol.

\noindent$\mathbf{Hybrid_4:}$ Same as $\mathbf{Hybrid_3}$ except that for $x\in X_n \setminus I$, $\sdv_1$ replaces the OPRF value of $x$ by $\ell_2$-bit random string. This hybrid is computationally indistinguishable from $\mathbf{Hybrid_4}$ by the $t$ party $d$-Hamming correlation robustness. For each $x\in X_n \setminus I$, let $v = F_k(H_1(x)) $, $a = A_1[v[1]] || \dots || A_w[v[w]]$ and $b_i = D^i_1[v[1]]||\dots||D^i_w[v[w]]$ for $i\in[n-1]$. We can compute that $C^n_1[v[1]] || \dots || C^n_w[v[w]] = a \oplus[b_1 \cdot s_2]\oplus\dots\oplus[b_{n-1} \cdot s_n]$. Since $||b_i||_H \geq d$ and $s_{i+1}$ is randomly sampled and unknown to $P_1$ for $i\in[n-1]$, the OPRF value of $x$ which is $H_2(C^n_1[v[1]] || \dots || C^n_w[v[w]])$ is pseudorandom for $P_1$ by the $t$ party $d$-Hamming correlation robustness of $H_2$.

\noindent$\mathbf{Hybrid_5:}$ Same as $\mathbf{Hybrid_4}$ except that the protocol does not abort. From the above discussion, it is evident that the probability of a protocol abort is negligible. This hybrid is $P_1$'s view simulated by $\mathcal{S}_1$.
\end{proof}
\section{Experiments}
\begin{table*}[htb]
\centering
\caption{Running time (in millisecond) of the MPSI protocols in LAN and WAN settings.}
\label{tab:LAN+WAN}
\renewcommand\arraystretch{1.2}
\resizebox{\linewidth}{!}{
\setlength{\tabcolsep}{1mm}{
\begin{tabular}{|c|c|cccccccccccccc|}
\hline
 & \multicolumn{1}{c|}{}& \multicolumn{14}{c|}{\textbf{Total running time for different parties(ms)}}\\ \cline{3-16} 
 & \multicolumn{1}{c|}{}& \multicolumn{7}{c|}{\textbf{LAN (5000M, 0.2ms RTT)}}& \multicolumn{7}{c|}{\textbf{WAN (200M, 2ms RTT)}}\\ \cline{3-16} 
 \multirow{-3}{*}{\textbf{Set size}} & \multicolumn{1}{c|}{\multirow{-3}{*}{\textbf{Protocol}}} & \multicolumn{1}{c|}{4} & \multicolumn{1}{c|}{6} & \multicolumn{1}{c|}{8} & \multicolumn{1}{c|}{10} & \multicolumn{1}{c|}{12} & \multicolumn{1}{c|}{14} & \multicolumn{1}{c|}{15} & \multicolumn{1}{c|}{4} & \multicolumn{1}{c|}{6} & \multicolumn{1}{c|}{8} & \multicolumn{1}{c|}{10} & \multicolumn{1}{c|}{12} & \multicolumn{1}{c|}{14} & 15 \\ \hline
 & KMPRT~\cite{kolesnikov2017practical} & \multicolumn{1}{c|}{364} & \multicolumn{1}{c|}{499} & \multicolumn{1}{c|}{864} & \multicolumn{1}{c|}{1140} & \multicolumn{1}{c|}{1524} & \multicolumn{1}{c|}{1986} & \multicolumn{1}{c|}{2268} & \multicolumn{1}{c|}{947} & \multicolumn{1}{c|}{2153} & \multicolumn{1}{c|}{4046} & \multicolumn{1}{c|}{6629} & \multicolumn{1}{c|}{9881} & \multicolumn{1}{c|}{13788} & 16009 \\ \cline{2-16} 
 & KMPRT(AUG)~\cite{kolesnikov2017practical} & \multicolumn{1}{c|}{400} & \multicolumn{1}{c|}{431} & \multicolumn{1}{c|}{482} & \multicolumn{1}{c|}{499} & \multicolumn{1}{c|}{589} & \multicolumn{1}{c|}{618} & \multicolumn{1}{c|}{625} & \multicolumn{1}{c|}{1592} & \multicolumn{1}{c|}{2274} & \multicolumn{1}{c|}{3092} & \multicolumn{1}{c|}{3880} & \multicolumn{1}{c|}{4679} & \multicolumn{1}{c|}{5498} & 5949 \\ \cline{2-16} 
 & NTY~\cite{nevo2021simple} & \multicolumn{1}{c|}{\cellcolor[HTML]{4E83FD}232} & \multicolumn{1}{c|}{\cellcolor[HTML]{4E83FD}238} & \multicolumn{1}{c|}{\cellcolor[HTML]{4E83FD}198} & \multicolumn{1}{c|}{\cellcolor[HTML]{4E83FD}260} & \multicolumn{1}{c|}{\cellcolor[HTML]{4E83FD}297} & \multicolumn{1}{c|}{\cellcolor[HTML]{4E83FD}320} & \multicolumn{1}{c|}{\cellcolor[HTML]{4E83FD}395} & \multicolumn{1}{c|}{\cellcolor[HTML]{4E83FD}365} & \multicolumn{1}{c|}{557} & \multicolumn{1}{c|}{728} & \multicolumn{1}{c|}{893} & \multicolumn{1}{c|}{1114} & \multicolumn{1}{c|}{1345} & 1434 \\ \cline{2-16} 
 \multirow{-4}{*}{$2^{12}$} & Ours & \multicolumn{1}{c|}{363} & \multicolumn{1}{c|}{374} & \multicolumn{1}{c|}{394} & \multicolumn{1}{c|}{466} & \multicolumn{1}{c|}{457} & \multicolumn{1}{c|}{465} & \multicolumn{1}{c|}{728} & \multicolumn{1}{c|}{366} & \multicolumn{1}{c|}{\cellcolor[HTML]{4E83FD}417} & \multicolumn{1}{c|}{\cellcolor[HTML]{4E83FD}493} & \multicolumn{1}{c|}{\cellcolor[HTML]{4E83FD}595} & \multicolumn{1}{c|}{\cellcolor[HTML]{4E83FD}660} & \multicolumn{1}{c|}{\cellcolor[HTML]{4E83FD}796} & \cellcolor[HTML]{4E83FD}673 \\ \hline
 & KMPRT~\cite{kolesnikov2017practical} & \multicolumn{1}{c|}{551} & \multicolumn{1}{c|}{885} & \multicolumn{1}{c|}{1435} & \multicolumn{1}{c|}{1853} & \multicolumn{1}{c|}{2488} & \multicolumn{1}{c|}{3132} & \multicolumn{1}{c|}{3602} & \multicolumn{1}{c|}{1831} & \multicolumn{1}{c|}{4537} & \multicolumn{1}{c|}{8598} & \multicolumn{1}{c|}{14061} & \multicolumn{1}{c|}{20909} & \multicolumn{1}{c|}{29186} & 33839 \\ \cline{2-16} 
 & KMPRT(AUG)~\cite{kolesnikov2017practical} & \multicolumn{1}{c|}{681} & \multicolumn{1}{c|}{802} & \multicolumn{1}{c|}{833} & \multicolumn{1}{c|}{915} & \multicolumn{1}{c|}{1037} & \multicolumn{1}{c|}{1141} & \multicolumn{1}{c|}{1219} & \multicolumn{1}{c|}{3417} & \multicolumn{1}{c|}{5197} & \multicolumn{1}{c|}{6689} & \multicolumn{1}{c|}{8488} & \multicolumn{1}{c|}{10231} & \multicolumn{1}{c|}{12072} & 12937 \\ \cline{2-16} 
 & NTY~\cite{nevo2021simple} & \multicolumn{1}{c|}{\cellcolor[HTML]{4E83FD}229} & \multicolumn{1}{c|}{\cellcolor[HTML]{4E83FD}294} & \multicolumn{1}{c|}{\cellcolor[HTML]{4E83FD}323} & \multicolumn{1}{c|}{\cellcolor[HTML]{4E83FD}339} & \multicolumn{1}{c|}{\cellcolor[HTML]{4E83FD}414} & \multicolumn{1}{c|}{\cellcolor[HTML]{4E83FD}427} & \multicolumn{1}{c|}{\cellcolor[HTML]{4E83FD}516} & \multicolumn{1}{c|}{599} & \multicolumn{1}{c|}{1041} & \multicolumn{1}{c|}{1355} & \multicolumn{1}{c|}{1726} & \multicolumn{1}{c|}{2122} & \multicolumn{1}{c|}{2645} & 2754 \\ \cline{2-16} 
\multirow{-4}{*}{$2^{13}$} & Ours & \multicolumn{1}{c|}{395} & \multicolumn{1}{c|}{453} & \multicolumn{1}{c|}{428} & \multicolumn{1}{c|}{527} & \multicolumn{1}{c|}{538} & \multicolumn{1}{c|}{524} & \multicolumn{1}{c|}{585} & \multicolumn{1}{c|}{\cellcolor[HTML]{4E83FD}439} & \multicolumn{1}{c|}{\cellcolor[HTML]{4E83FD}575} & \multicolumn{1}{c|}{\cellcolor[HTML]{4E83FD}669} & \multicolumn{1}{c|}{\cellcolor[HTML]{4E83FD}790} & \multicolumn{1}{c|}{\cellcolor[HTML]{4E83FD}886} & \multicolumn{1}{c|}{\cellcolor[HTML]{4E83FD}991} & \cellcolor[HTML]{4E83FD}1315 \\ \hline
 & KMPRT~\cite{kolesnikov2017practical} & \multicolumn{1}{c|}{996} & \multicolumn{1}{c|}{1521} & \multicolumn{1}{c|}{2326} & \multicolumn{1}{c|}{3500} & \multicolumn{1}{c|}{4525} & \multicolumn{1}{c|}{5715} & \multicolumn{1}{c|}{6327} & \multicolumn{1}{c|}{3547} & \multicolumn{1}{c|}{8641} & \multicolumn{1}{c|}{16809} & \multicolumn{1}{c|}{27696} & \multicolumn{1}{c|}{41296} & \multicolumn{1}{c|}{57663} & 66818 \\ \cline{2-16} 
 & KMPRT(AUG)~\cite{kolesnikov2017practical} & \multicolumn{1}{c|}{1288} & \multicolumn{1}{c|}{1434} & \multicolumn{1}{c|}{1585} & \multicolumn{1}{c|}{1742} & \multicolumn{1}{c|}{1893} & \multicolumn{1}{c|}{2088} & \multicolumn{1}{c|}{2171} & \multicolumn{1}{c|}{6194} & \multicolumn{1}{c|}{9663} & \multicolumn{1}{c|}{13252} & \multicolumn{1}{c|}{16774} & \multicolumn{1}{c|}{20298} & \multicolumn{1}{c|}{23850} & 25753 \\ \cline{2-16} 
 & NTY~\cite{nevo2021simple} & \multicolumn{1}{c|}{\cellcolor[HTML]{4E83FD}334} & \multicolumn{1}{c|}{\cellcolor[HTML]{4E83FD}508} & \multicolumn{1}{c|}{\cellcolor[HTML]{4E83FD}542} & \multicolumn{1}{c|}{\cellcolor[HTML]{4E83FD}567} & \multicolumn{1}{c|}{632} & \multicolumn{1}{c|}{753} & \multicolumn{1}{c|}{\cellcolor[HTML]{4E83FD}686} & \multicolumn{1}{c|}{1086} & \multicolumn{1}{c|}{1751} & \multicolumn{1}{c|}{2541} & \multicolumn{1}{c|}{3295} & \multicolumn{1}{c|}{4055} & \multicolumn{1}{c|}{4964} & 5388 \\ \cline{2-16} 
\multirow{-4}{*}{$2^{14}$} & Ours & \multicolumn{1}{c|}{475} & \multicolumn{1}{c|}{548} & \multicolumn{1}{c|}{559} & \multicolumn{1}{c|}{600} & \multicolumn{1}{c|}{\cellcolor[HTML]{4E83FD}594} & \multicolumn{1}{c|}{\cellcolor[HTML]{4E83FD}632} & \multicolumn{1}{c|}{969} & \multicolumn{1}{c|}{\cellcolor[HTML]{4E83FD}646} & \multicolumn{1}{c|}{\cellcolor[HTML]{4E83FD}868} & \multicolumn{1}{c|}{\cellcolor[HTML]{4E83FD}988} & \multicolumn{1}{c|}{\cellcolor[HTML]{4E83FD}1221}  & \multicolumn{1}{c|}{\cellcolor[HTML]{4E83FD}2160}  & \multicolumn{1}{c|}{\cellcolor[HTML]{4E83FD}2660}  & \cellcolor[HTML]{4E83FD}3339  \\ \hline
 & KMPRT~\cite{kolesnikov2017practical} & \multicolumn{1}{c|}{1735} & \multicolumn{1}{c|}{3077} & \multicolumn{1}{c|}{4600} & \multicolumn{1}{c|}{6597} & \multicolumn{1}{c|}{8806} & \multicolumn{1}{c|}{10552} & \multicolumn{1}{c|}{11882} & \multicolumn{1}{c|}{6869} & \multicolumn{1}{c|}{17035} & \multicolumn{1}{c|}{33101} & \multicolumn{1}{c|}{54600} & \multicolumn{1}{c|}{81490} & \multicolumn{1}{c|}{113746} & 131870 \\ \cline{2-16} 
 & KMPRT(AUG)~\cite{kolesnikov2017practical} & \multicolumn{1}{c|}{2645} & \multicolumn{1}{c|}{2882} & \multicolumn{1}{c|}{3259} & \multicolumn{1}{c|}{3568} & \multicolumn{1}{c|}{3730} & \multicolumn{1}{c|}{4236} & \multicolumn{1}{c|}{4273} & \multicolumn{1}{c|}{13516} & \multicolumn{1}{c|}{20494} & \multicolumn{1}{c|}{26359} & \multicolumn{1}{c|}{33339} & \multicolumn{1}{c|}{40494} & \multicolumn{1}{c|}{47818} & 51368 \\ \cline{2-16} 
 & NTY~\cite{nevo2021simple} & \multicolumn{1}{c|}{\cellcolor[HTML]{4E83FD}593} & \multicolumn{1}{c|}{819} & \multicolumn{1}{c|}{919} & \multicolumn{1}{c|}{960} & \multicolumn{1}{c|}{\cellcolor[HTML]{4E83FD}1134}  & \multicolumn{1}{c|}{\cellcolor[HTML]{4E83FD}1157}  & \multicolumn{1}{c|}{\cellcolor[HTML]{4E83FD}1255}  & \multicolumn{1}{c|}{2103} & \multicolumn{1}{c|}{3363} & \multicolumn{1}{c|}{4935} & \multicolumn{1}{c|}{6447} & \multicolumn{1}{c|}{7997} & \multicolumn{1}{c|}{9572} & 10398 \\ \cline{2-16} 
 \multirow{-4}{*}{$2^{15}$} & Ours & \multicolumn{1}{c|}{635} & \multicolumn{1}{c|}{\cellcolor[HTML]{4E83FD}663} & \multicolumn{1}{c|}{\cellcolor[HTML]{4E83FD}701} & \multicolumn{1}{c|}{\cellcolor[HTML]{4E83FD}764} & \multicolumn{1}{c|}{1521} & \multicolumn{1}{c|}{2058} & \multicolumn{1}{c|}{2070} & \multicolumn{1}{c|}{\cellcolor[HTML]{4E83FD}952} & \multicolumn{1}{c|}{\cellcolor[HTML]{4E83FD}1330}  & \multicolumn{1}{c|}{\cellcolor[HTML]{4E83FD}1868}  & \multicolumn{1}{c|}{\cellcolor[HTML]{4E83FD}2164}  & \multicolumn{1}{c|}{\cellcolor[HTML]{4E83FD}3356}  & \multicolumn{1}{c|}{\cellcolor[HTML]{4E83FD}3362}  & \cellcolor[HTML]{4E83FD}4801  \\ \hline
 & KMPRT~\cite{kolesnikov2017practical} & \multicolumn{1}{c|}{3087} & \multicolumn{1}{c|}{5982} & \multicolumn{1}{c|}{9854} & \multicolumn{1}{c|}{12995} & \multicolumn{1}{c|}{16776} & \multicolumn{1}{c|}{20642} & \multicolumn{1}{c|}{23348} & \multicolumn{1}{c|}{13788} & \multicolumn{1}{c|}{34382} & \multicolumn{1}{c|}{66066} & \multicolumn{1}{c|}{108869} & \multicolumn{1}{c|}{162357} & \multicolumn{1}{c|}{226912} & 262982 \\ \cline{2-16} 
 & KMPRT(AUG)~\cite{kolesnikov2017practical} & \multicolumn{1}{c|}{5482} & \multicolumn{1}{c|}{5929} & \multicolumn{1}{c|}{6477} & \multicolumn{1}{c|}{7285} & \multicolumn{1}{c|}{7543} & \multicolumn{1}{c|}{8235} & \multicolumn{1}{c|}{8788} & \multicolumn{1}{c|}{26920} & \multicolumn{1}{c|}{38385} & \multicolumn{1}{c|}{52767} & \multicolumn{1}{c|}{66796} & \multicolumn{1}{c|}{81011} & \multicolumn{1}{c|}{95177} & 101991 \\ \cline{2-16} 
 & NTY~\cite{nevo2021simple} & \multicolumn{1}{c|}{\cellcolor[HTML]{4E83FD}1083}  & \multicolumn{1}{c|}{\cellcolor[HTML]{4E83FD}1348}  & \multicolumn{1}{c|}{\cellcolor[HTML]{4E83FD}1620}  & \multicolumn{1}{c|}{\cellcolor[HTML]{4E83FD}1769}  & \multicolumn{1}{c|}{\cellcolor[HTML]{4E83FD}1936}  & \multicolumn{1}{c|}{\cellcolor[HTML]{4E83FD}2225}  & \multicolumn{1}{c|}{\cellcolor[HTML]{4E83FD}2308}  & \multicolumn{1}{c|}{3993} & \multicolumn{1}{c|}{6700} & \multicolumn{1}{c|}{9717} & \multicolumn{1}{c|}{12777} & \multicolumn{1}{c|}{15821} & \multicolumn{1}{c|}{19019} & 20559 \\ \cline{2-16} 
 \multirow{-4}{*}{$2^{16}$} & Ours & \multicolumn{1}{c|}{1003} & \multicolumn{1}{c|}{1037} & \multicolumn{1}{c|}{1070} & \multicolumn{1}{c|}{3646} & \multicolumn{1}{c|}{3435} & \multicolumn{1}{c|}{3522} & \multicolumn{1}{c|}{4242} & \multicolumn{1}{c|}{\cellcolor[HTML]{4E83FD}1633}  & \multicolumn{1}{c|}{\cellcolor[HTML]{4E83FD}2409}  & \multicolumn{1}{c|}{\cellcolor[HTML]{4E83FD}4601}  & \multicolumn{1}{c|}{\cellcolor[HTML]{4E83FD}4259}  & \multicolumn{1}{c|}{\cellcolor[HTML]{4E83FD}6054}  & \multicolumn{1}{c|}{\cellcolor[HTML]{4E83FD}8599}  & \cellcolor[HTML]{4E83FD}9255  \\ \hline
 & KMPRT~\cite{kolesnikov2017practical} & \multicolumn{1}{c|}{6549} & \multicolumn{1}{c|}{11218} & \multicolumn{1}{c|}{17340} & \multicolumn{1}{c|}{24655} & \multicolumn{1}{c|}{32832} & \multicolumn{1}{c|}{41030} & \multicolumn{1}{c|}{44658} & \multicolumn{1}{c|}{29658} & \multicolumn{1}{c|}{75154} & \multicolumn{1}{c|}{144507} & \multicolumn{1}{c|}{238440} & \multicolumn{1}{c|}{355742} & \multicolumn{1}{c|}{496675} & 576085 \\ \cline{2-16} 
 & KMPRT(AUG)~\cite{kolesnikov2017practical}  & \multicolumn{1}{c|}{10664} & \multicolumn{1}{c|}{12934} & \multicolumn{1}{c|}{13655} & \multicolumn{1}{c|}{15498} & \multicolumn{1}{c|}{17592} & \multicolumn{1}{c|}{18354} & \multicolumn{1}{c|}{19229} & \multicolumn{1}{c|}{59653} & \multicolumn{1}{c|}{85886} & \multicolumn{1}{c|}{116037} & \multicolumn{1}{c|}{147080} & \multicolumn{1}{c|}{178280} & \multicolumn{1}{c|}{209563} & 225231 \\ \cline{2-16} 
 & NTY~\cite{nevo2021simple} & \multicolumn{1}{c|}{\cellcolor[HTML]{4E83FD}2133}  & \multicolumn{1}{c|}{\cellcolor[HTML]{4E83FD}2505}  & \multicolumn{1}{c|}{\cellcolor[HTML]{4E83FD}3041}  & \multicolumn{1}{c|}{\cellcolor[HTML]{4E83FD}3225}  & \multicolumn{1}{c|}{\cellcolor[HTML]{4E83FD}3798}  & \multicolumn{1}{c|}{\cellcolor[HTML]{4E83FD}4166}  & \multicolumn{1}{c|}{\cellcolor[HTML]{4E83FD}4474}  & \multicolumn{1}{c|}{8628} & \multicolumn{1}{c|}{14166} & \multicolumn{1}{c|}{20867} & \multicolumn{1}{c|}{27532} & \multicolumn{1}{c|}{34252} & \multicolumn{1}{c|}{40943} & 44400 \\ \cline{2-16} 
 \multirow{-4}{*}{$2^{17}$} & Ours & \multicolumn{1}{c|}{2814} & \multicolumn{1}{c|}{3598} & \multicolumn{1}{c|}{3672} & \multicolumn{1}{c|}{4582} & \multicolumn{1}{c|}{5293} & \multicolumn{1}{c|}{5585} & \multicolumn{1}{c|}{5465} & \multicolumn{1}{c|}{\cellcolor[HTML]{4E83FD}6367}  & \multicolumn{1}{c|}{\cellcolor[HTML]{4E83FD}7335}  & \multicolumn{1}{c|}{\cellcolor[HTML]{4E83FD}8300}  & \multicolumn{1}{c|}{\cellcolor[HTML]{4E83FD}9809}  & \multicolumn{1}{c|}{\cellcolor[HTML]{4E83FD}12753} & \multicolumn{1}{c|}{\cellcolor[HTML]{4E83FD}14148} & \cellcolor[HTML]{4E83FD}14530 \\ \hline
 & KMPRT~\cite{kolesnikov2017practical} & \multicolumn{1}{c|}{12338} & \multicolumn{1}{c|}{21912} & \multicolumn{1}{c|}{35974} & \multicolumn{1}{c|}{49526} & \multicolumn{1}{c|}{63299} & \multicolumn{1}{c|}{81296} & \multicolumn{1}{c|}{91360} & \multicolumn{1}{c|}{60064} & \multicolumn{1}{c|}{148683} & \multicolumn{1}{c|}{288943} & \multicolumn{1}{c|}{476566} & \multicolumn{1}{c|}{711234} & \multicolumn{1}{c|}{992698} & 1152083  \\ \cline{2-16} 
 & KMPRT(AUG)~\cite{kolesnikov2017practical}  & \multicolumn{1}{c|}{24767} & \multicolumn{1}{c|}{26423} & \multicolumn{1}{c|}{29833} & \multicolumn{1}{c|}{32954} & \multicolumn{1}{c|}{35867} & \multicolumn{1}{c|}{39160} & \multicolumn{1}{c|}{39499} & \multicolumn{1}{c|}{119125} & \multicolumn{1}{c|}{180798} & \multicolumn{1}{c|}{232095} & \multicolumn{1}{c|}{295276} & \multicolumn{1}{c|}{356877} & \multicolumn{1}{c|}{423846} & 452325 \\ \cline{2-16} 
 & NTY~\cite{nevo2021simple} & \multicolumn{1}{c|}{4425} & \multicolumn{1}{c|}{5161} & \multicolumn{1}{c|}{5770} & \multicolumn{1}{c|}{6540} & \multicolumn{1}{c|}{7831} & \multicolumn{1}{c|}{8930} & \multicolumn{1}{c|}{9634} & \multicolumn{1}{c|}{17246} & \multicolumn{1}{c|}{28314} & \multicolumn{1}{c|}{41741} & \multicolumn{1}{c|}{54827} & \multicolumn{1}{c|}{68289} & \multicolumn{1}{c|}{81717} & 88472 \\ \cline{2-16} 
 \multirow{-4}{*}{$2^{18}$} & Ours & \multicolumn{1}{c|}{\cellcolor[HTML]{4E83FD}3205}  & \multicolumn{1}{c|}{\cellcolor[HTML]{4E83FD}4133}  & \multicolumn{1}{c|}{\cellcolor[HTML]{4E83FD}4480}  & \multicolumn{1}{c|}{\cellcolor[HTML]{4E83FD}6136}  & \multicolumn{1}{c|}{\cellcolor[HTML]{4E83FD}6929}  & \multicolumn{1}{c|}{\cellcolor[HTML]{4E83FD}7902}  & \multicolumn{1}{c|}{\cellcolor[HTML]{4E83FD}8085}  & \multicolumn{1}{c|}{\cellcolor[HTML]{4E83FD}7503}  & \multicolumn{1}{c|}{\cellcolor[HTML]{4E83FD}12085} & \multicolumn{1}{c|}{\cellcolor[HTML]{4E83FD}13582} & \multicolumn{1}{c|}{\cellcolor[HTML]{4E83FD}18579} & \multicolumn{1}{c|}{\cellcolor[HTML]{4E83FD}22203} & \multicolumn{1}{c|}{\cellcolor[HTML]{4E83FD}24799} & \cellcolor[HTML]{4E83FD}27074 \\ \hline
 & KMPRT~\cite{kolesnikov2017practical} & \multicolumn{1}{c|}{27061} & \multicolumn{1}{c|}{42049} & \multicolumn{1}{c|}{70068} & \multicolumn{1}{c|}{98985} & \multicolumn{1}{c|}{131225} & \multicolumn{1}{c|}{164695} & \multicolumn{1}{c|}{180946} & \multicolumn{1}{c|}{122911} & \multicolumn{1}{c|}{301132} & \multicolumn{1}{c|}{578093} & \multicolumn{1}{c|}{952718} & \multicolumn{1}{c|}{1421702}  & \multicolumn{1}{c|}{1985126}  & 2301667  \\ \cline{2-16} 
 & KMPRT(AUG)~\cite{kolesnikov2017practical}  & \multicolumn{1}{c|}{50014} & \multicolumn{1}{c|}{64313} & \multicolumn{1}{c|}{65266} & \multicolumn{1}{c|}{67896} & \multicolumn{1}{c|}{75725} & \multicolumn{1}{c|}{77460} & \multicolumn{1}{c|}{86233} & \multicolumn{1}{c|}{246857} & \multicolumn{1}{c|}{361365} & \multicolumn{1}{c|}{488146} & \multicolumn{1}{c|}{618449} & \multicolumn{1}{c|}{746404} & \multicolumn{1}{c|}{849883} & 934572 \\ \cline{2-16} 
 & NTY~\cite{nevo2021simple} & \multicolumn{1}{c|}{9056} & \multicolumn{1}{c|}{10179} & \multicolumn{1}{c|}{11537} & \multicolumn{1}{c|}{13194} & \multicolumn{1}{c|}{15445} & \multicolumn{1}{c|}{17475} & \multicolumn{1}{c|}{18861} & \multicolumn{1}{c|}{34880} & \multicolumn{1}{c|}{57035} & \multicolumn{1}{c|}{83062} & \multicolumn{1}{c|}{109812} & \multicolumn{1}{c|}{136550} & \multicolumn{1}{c|}{165247} & 176370 \\ \cline{2-16} 
 \multirow{-4}{*}{$2^{19}$} & Ours & \multicolumn{1}{c|}{\cellcolor[HTML]{4E83FD}8476}  & \multicolumn{1}{c|}{\cellcolor[HTML]{4E83FD}8540}  & \multicolumn{1}{c|}{\cellcolor[HTML]{4E83FD}8653}  & \multicolumn{1}{c|}{\cellcolor[HTML]{4E83FD}8698}  & \multicolumn{1}{c|}{\cellcolor[HTML]{4E83FD}8692}  & \multicolumn{1}{c|}{\cellcolor[HTML]{4E83FD}9476}  & \multicolumn{1}{c|}{\cellcolor[HTML]{4E83FD}9970}  & \multicolumn{1}{c|}{\cellcolor[HTML]{4E83FD}12193} & \multicolumn{1}{c|}{\cellcolor[HTML]{4E83FD}21440} & \multicolumn{1}{c|}{\cellcolor[HTML]{4E83FD}23439} & \multicolumn{1}{c|}{\cellcolor[HTML]{4E83FD}32219} & \multicolumn{1}{c|}{\cellcolor[HTML]{4E83FD}39165} & \multicolumn{1}{c|}{\cellcolor[HTML]{4E83FD}45423} & \cellcolor[HTML]{4E83FD}48800 \\ \hline
 & KMPRT~\cite{kolesnikov2017practical} & \multicolumn{1}{c|}{63815} & \multicolumn{1}{c|}{85631} & \multicolumn{1}{c|}{140710} & \multicolumn{1}{c|}{204628} & \multicolumn{1}{c|}{256021} & \multicolumn{1}{c|}{325696} & \multicolumn{1}{c|}{-}  & \multicolumn{1}{c|}{247219} & \multicolumn{1}{c|}{596205} & \multicolumn{1}{c|}{1159041}  & \multicolumn{1}{c|}{1908268}  & \multicolumn{1}{c|}{2845787}  & \multicolumn{1}{c|}{3974837}  & \multicolumn{1}{c|}{-}  \\ \cline{2-16} 
 & KMPRT(AUG)~\cite{kolesnikov2017practical}  & \multicolumn{1}{c|}{122619} & \multicolumn{1}{c|}{116997} & \multicolumn{1}{c|}{140912} & \multicolumn{1}{c|}{148125} & \multicolumn{1}{c|}{157234} & \multicolumn{1}{c|}{170809} & \multicolumn{1}{c|}{182509} & \multicolumn{1}{c|}{483092} & \multicolumn{1}{c|}{733914} & \multicolumn{1}{c|}{982592} & \multicolumn{1}{c|}{1248201}  & \multicolumn{1}{c|}{1501481}  & \multicolumn{1}{c|}{1722380}  & 1878569  \\ \cline{2-16} 
 & NTY~\cite{nevo2021simple} & \multicolumn{1}{c|}{19177} & \multicolumn{1}{c|}{21927} & \multicolumn{1}{c|}{25099} & \multicolumn{1}{c|}{30163} & \multicolumn{1}{c|}{33463} & \multicolumn{1}{c|}{36118} & \multicolumn{1}{c|}{38696} & \multicolumn{1}{c|}{71108} & \multicolumn{1}{c|}{123544} & \multicolumn{1}{c|}{167425} & \multicolumn{1}{c|}{219416} & \multicolumn{1}{c|}{276702} & \multicolumn{1}{c|}{339043} & 371039 \\ \cline{2-16} 
 \multirow{-4}{*}{$2^{20}$} & Ours & \multicolumn{1}{c|}{\cellcolor[HTML]{4E83FD}13109} & \multicolumn{1}{c|}{\cellcolor[HTML]{4E83FD}13859} & \multicolumn{1}{c|}{\cellcolor[HTML]{4E83FD}14170} & \multicolumn{1}{c|}{\cellcolor[HTML]{4E83FD}15284} & \multicolumn{1}{c|}{\cellcolor[HTML]{4E83FD}16755} & \multicolumn{1}{c|}{\cellcolor[HTML]{4E83FD}16893} & \multicolumn{1}{c|}{\cellcolor[HTML]{4E83FD}16753} & \multicolumn{1}{c|}{\cellcolor[HTML]{4E83FD}27000} & \multicolumn{1}{c|}{\cellcolor[HTML]{4E83FD}39202} & \multicolumn{1}{c|}{\cellcolor[HTML]{4E83FD}50541} & \multicolumn{1}{c|}{\cellcolor[HTML]{4E83FD}62909} & \multicolumn{1}{c|}{\cellcolor[HTML]{4E83FD}75512} & \multicolumn{1}{c|}{\cellcolor[HTML]{4E83FD}89505} & \cellcolor[HTML]{4E83FD}95911 \\ \hline
\end{tabular}}}
\end{table*}

\begin{figure*}[htb]
    \centering
    \subfigure[$|S|=2^{12}$]{\includegraphics[width=0.19\linewidth]{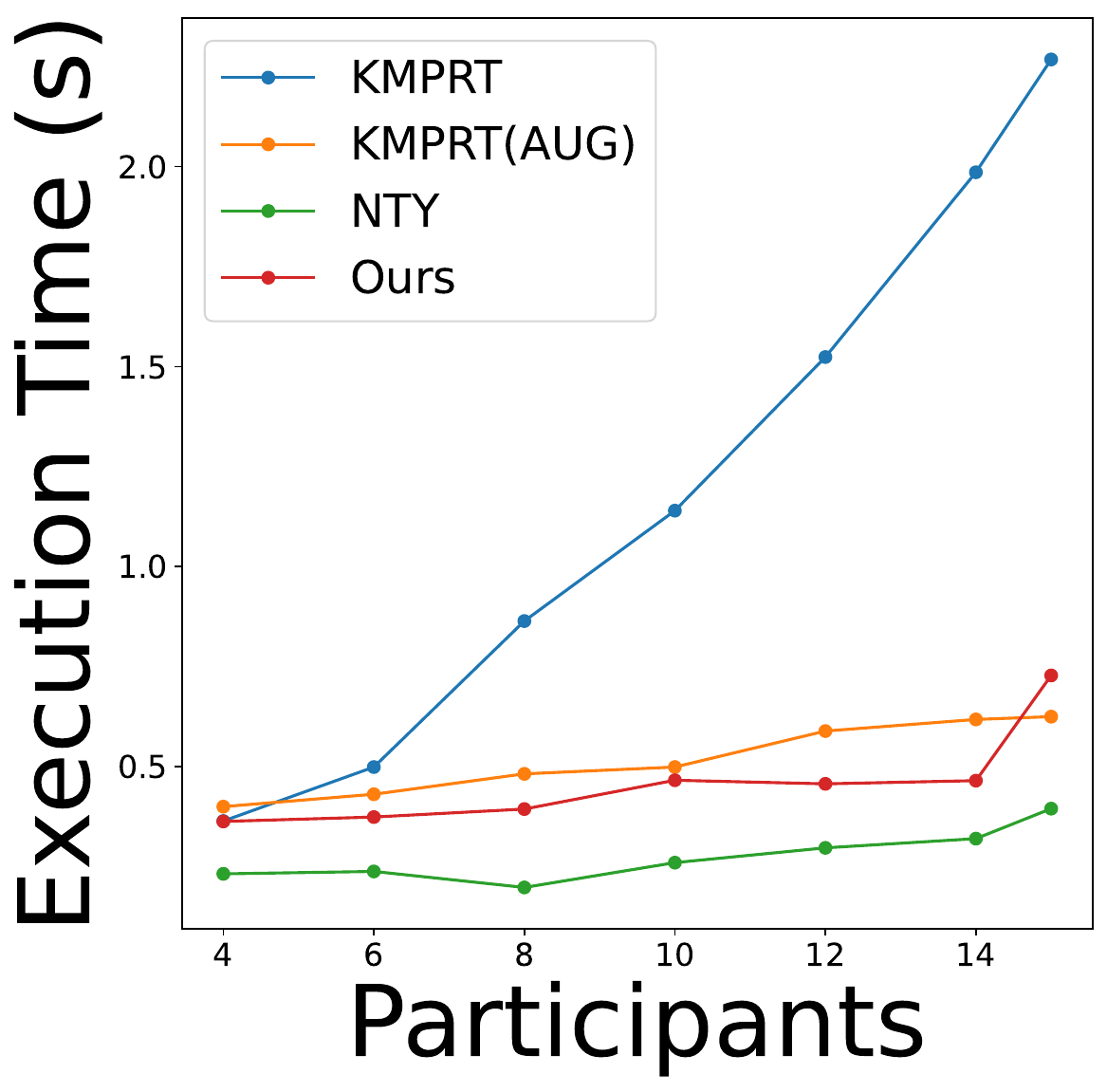}}\hspace{10pt}
    \subfigure[$|S|=2^{14}$]{\includegraphics[width=0.19\linewidth]{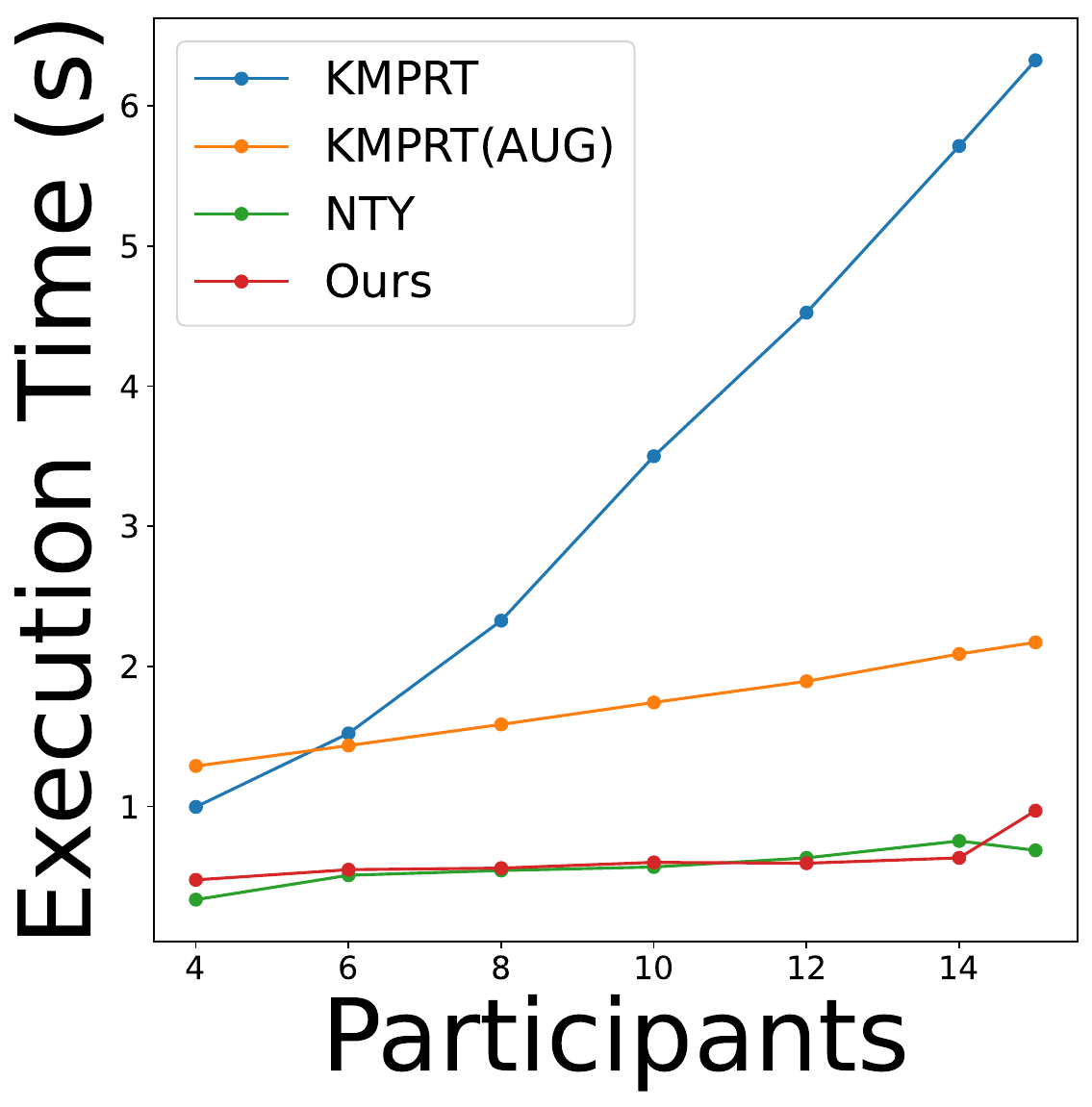}}
    \subfigure[$|S|=2^{16}$]{\includegraphics[width=0.19\linewidth]{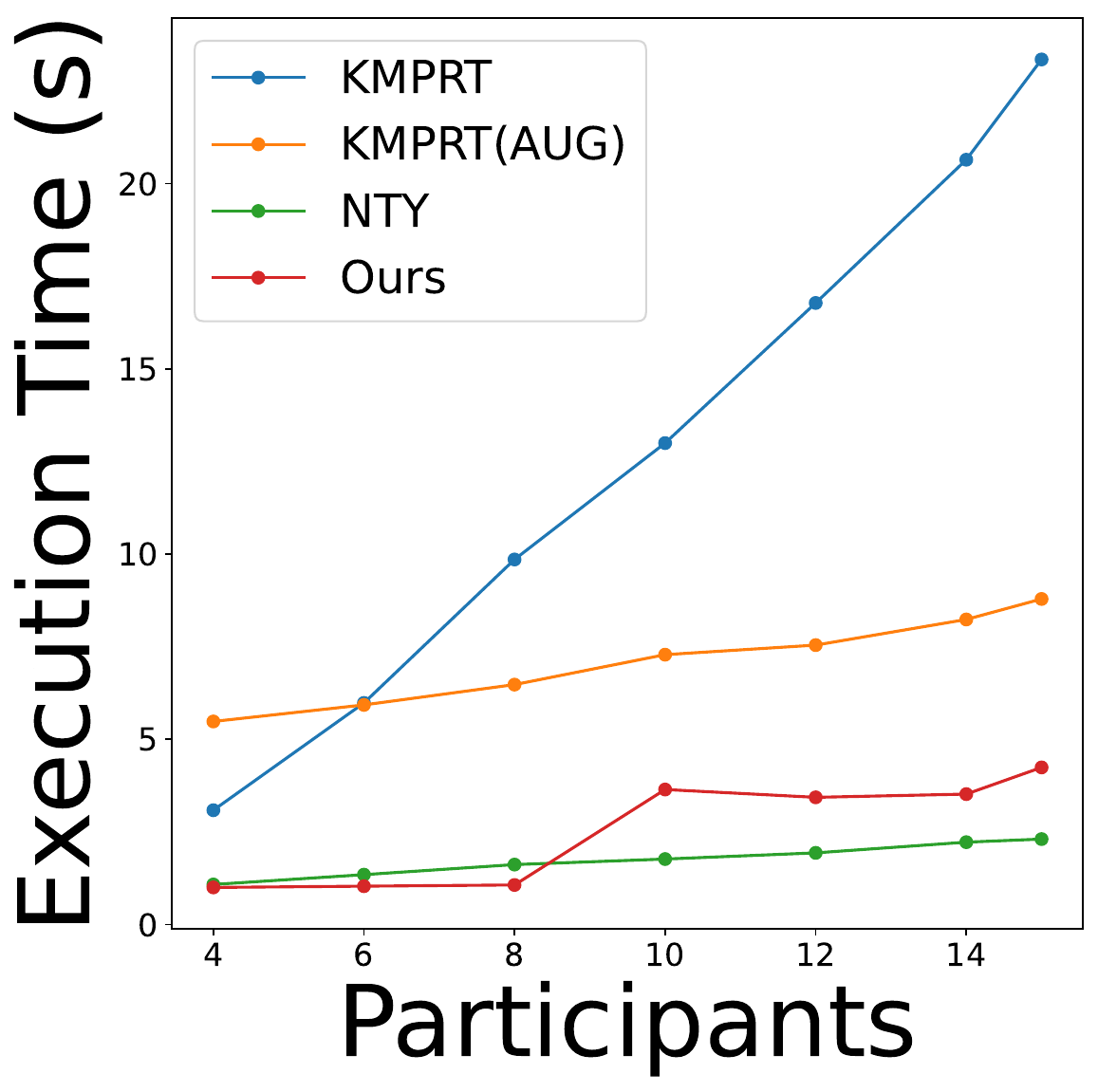}}
    \subfigure[$|S|=2^{18}$]{\includegraphics[width=0.19\linewidth]{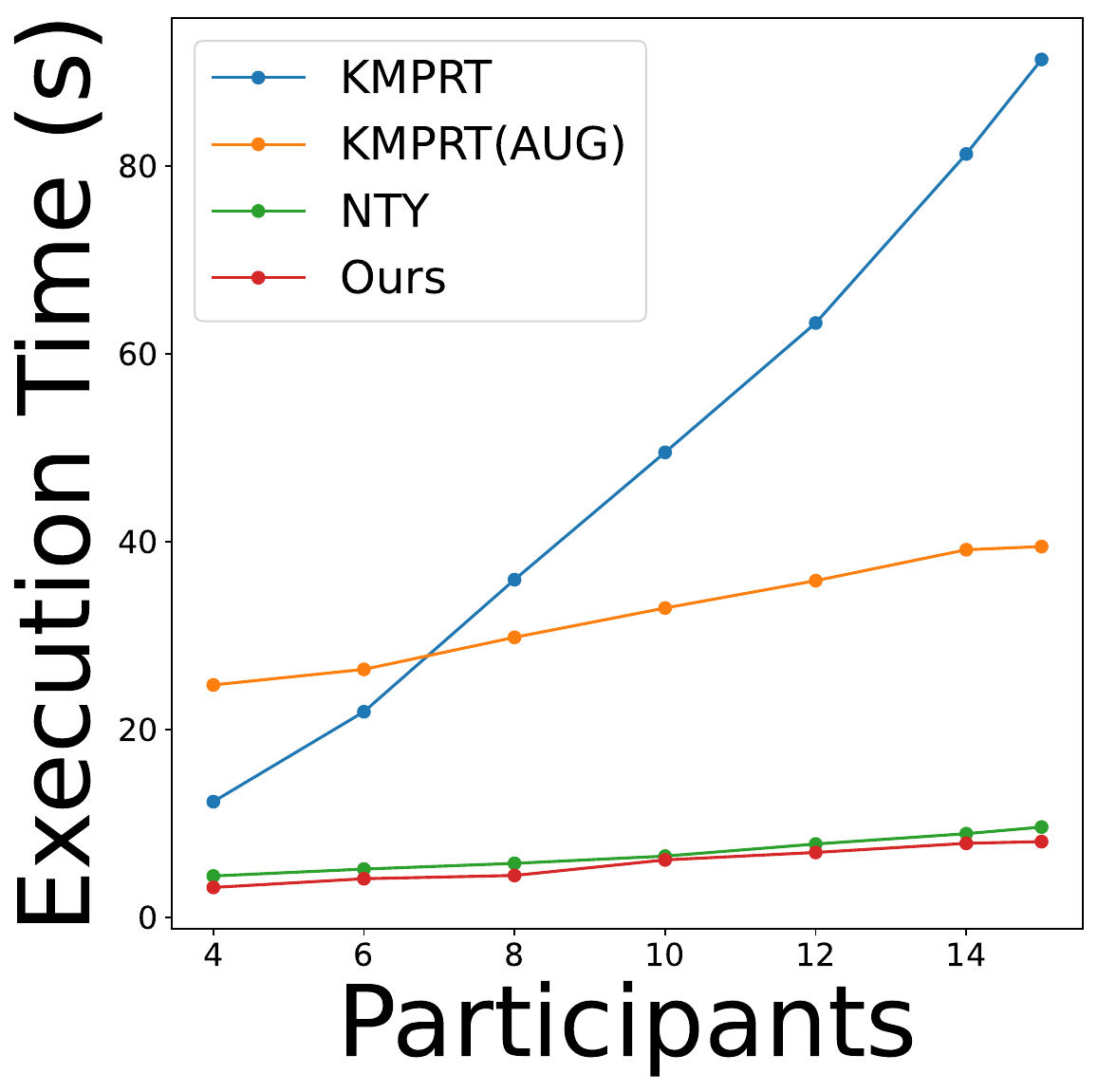}}
    \subfigure[$|S|=2^{20}$]{\includegraphics[width=0.19\linewidth]{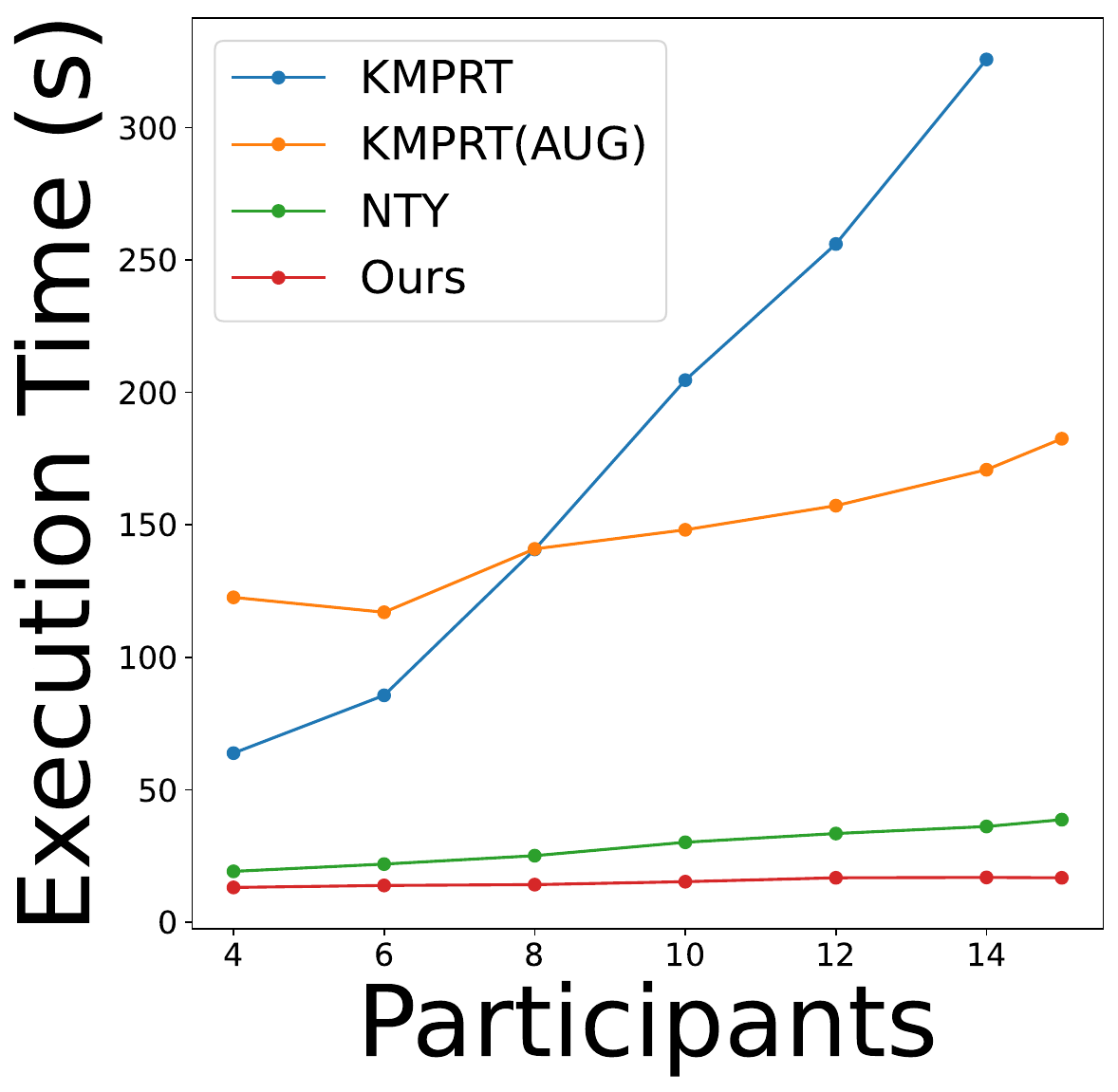}}
    \caption{Comparison of MPSI protocols under different data set sizes in LAN setting}
  \label{fig:lan}
\end{figure*}

\begin{figure*}[htb]
    \centering
    \subfigure[$|S|=2^{12}$]{\includegraphics[width=0.19\linewidth]{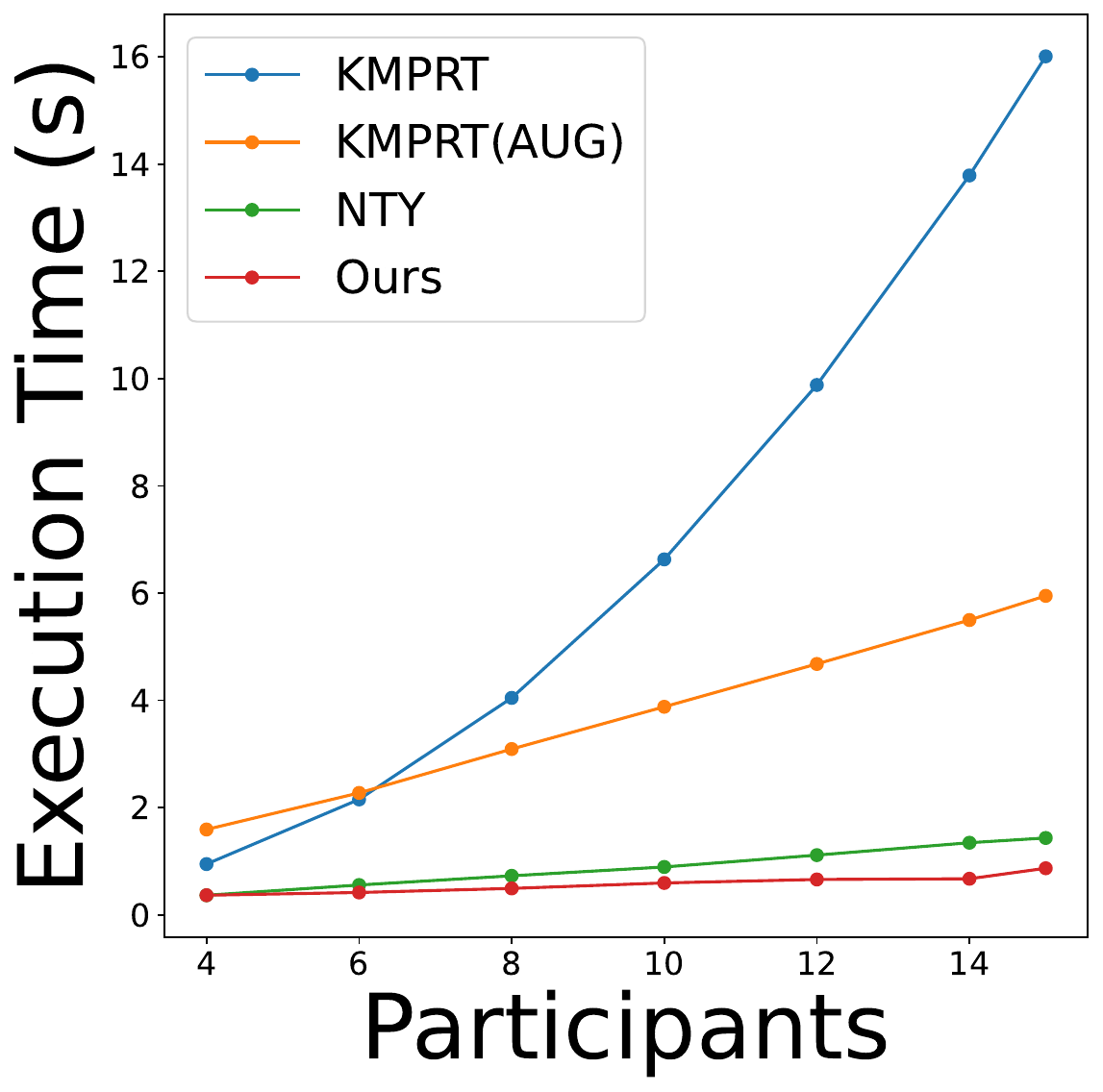}}
    \subfigure[$|S|=2^{14}$]{\includegraphics[width=0.19\linewidth]{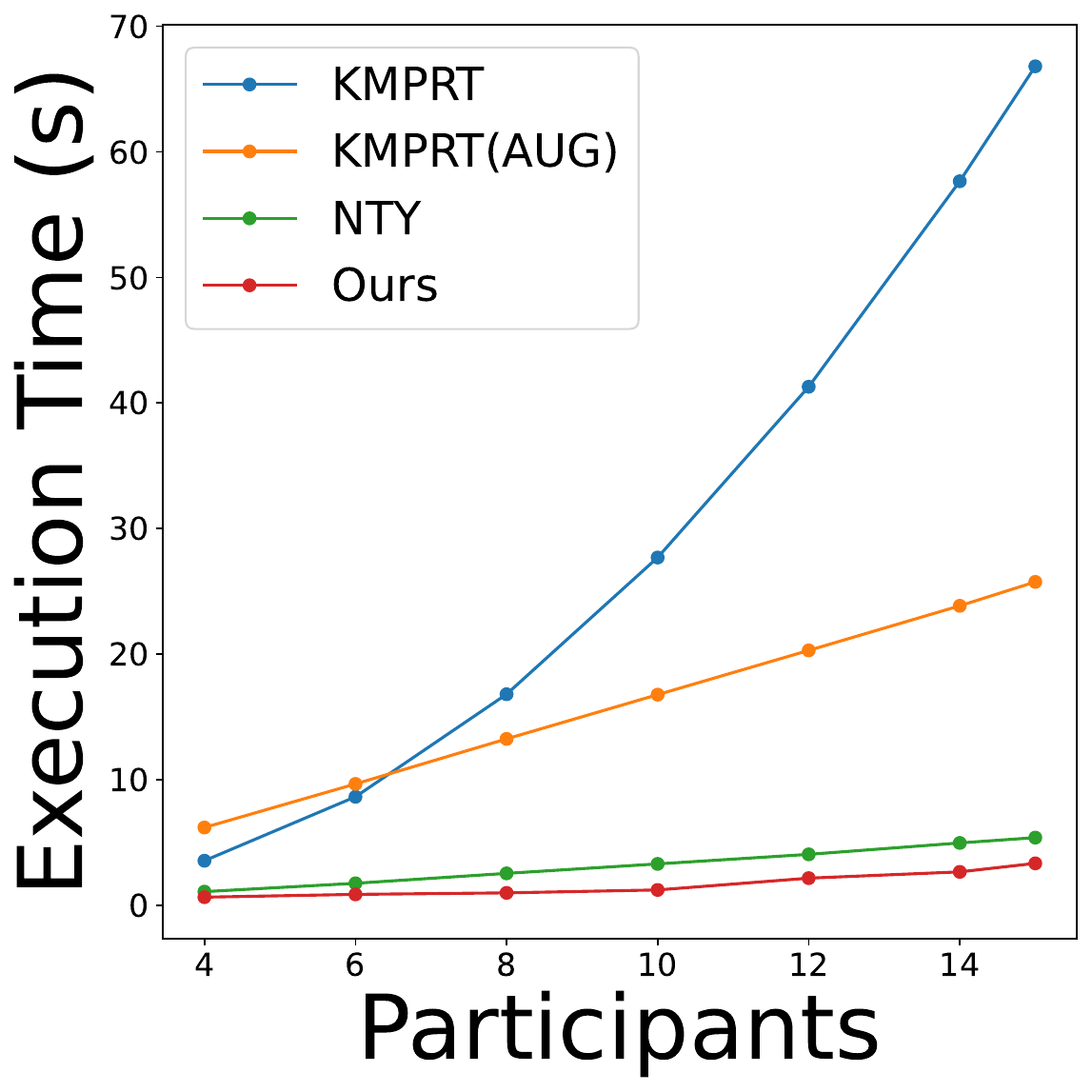}}
    \subfigure[$|S|=2^{16}$]{\includegraphics[width=0.19\linewidth]{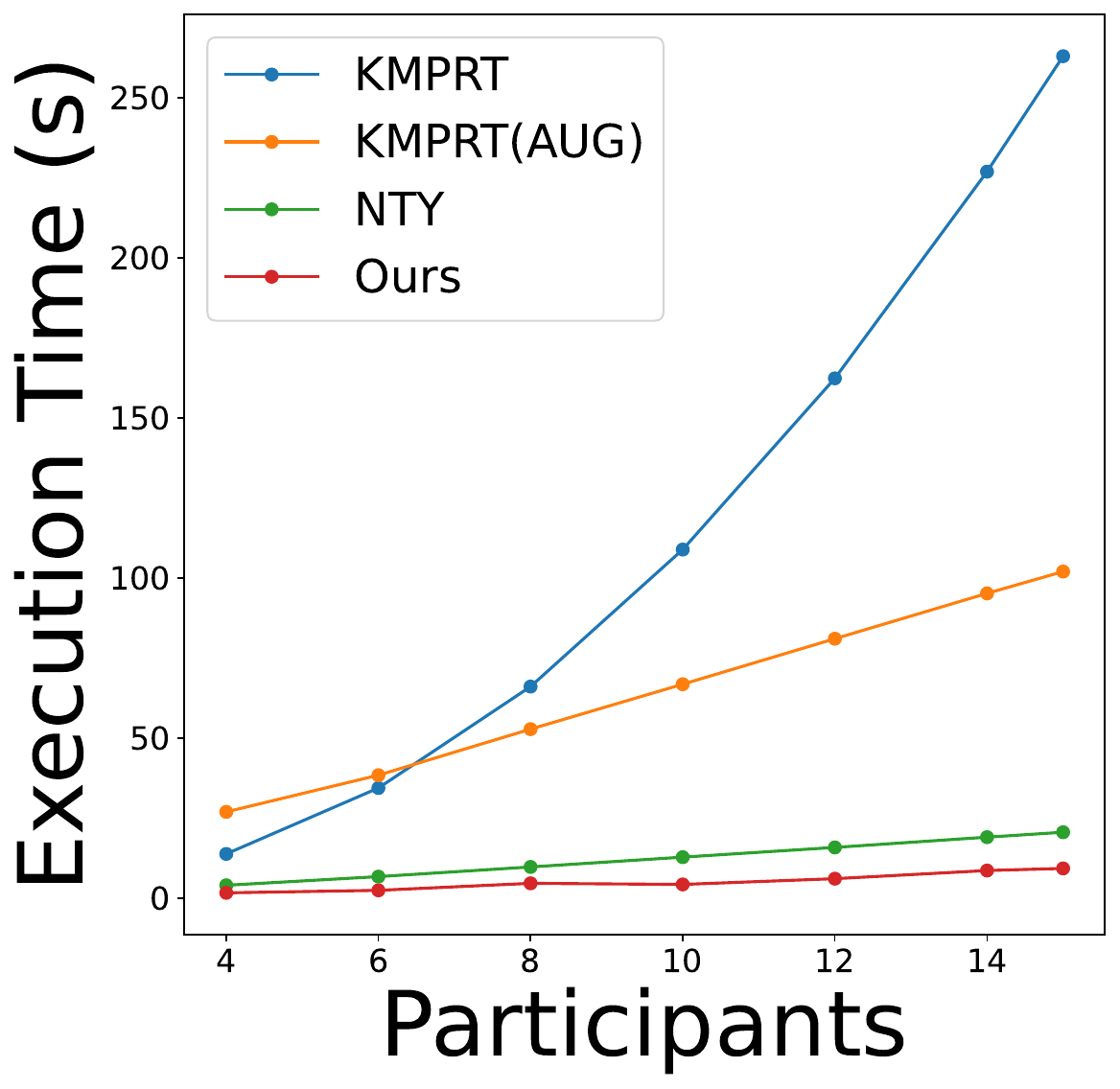}}
    \subfigure[$|S|=2^{18}$]{\includegraphics[width=0.19\linewidth]{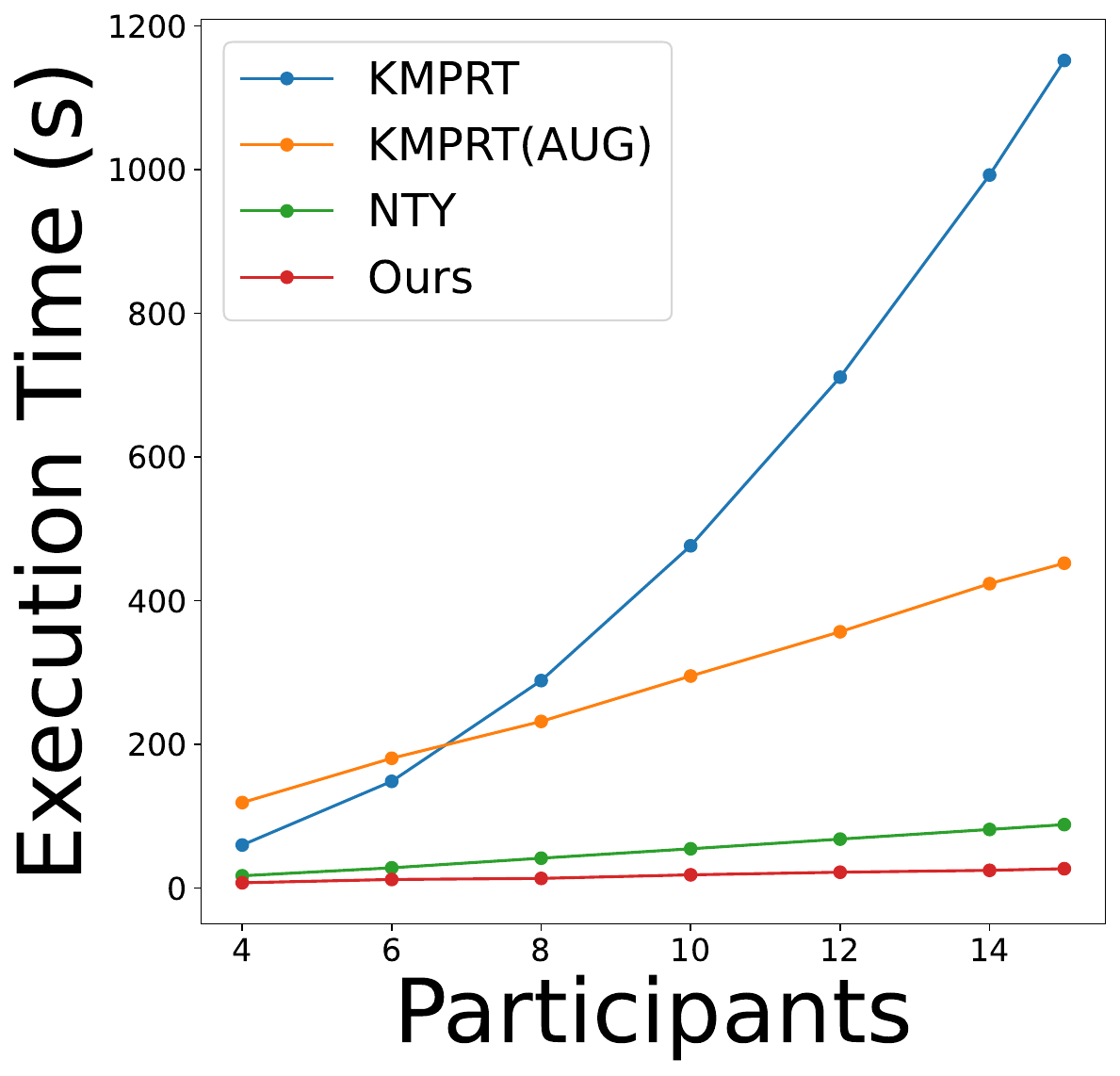}}
    \subfigure[$|S|=2^{20}$]{\includegraphics[width=0.19\linewidth]{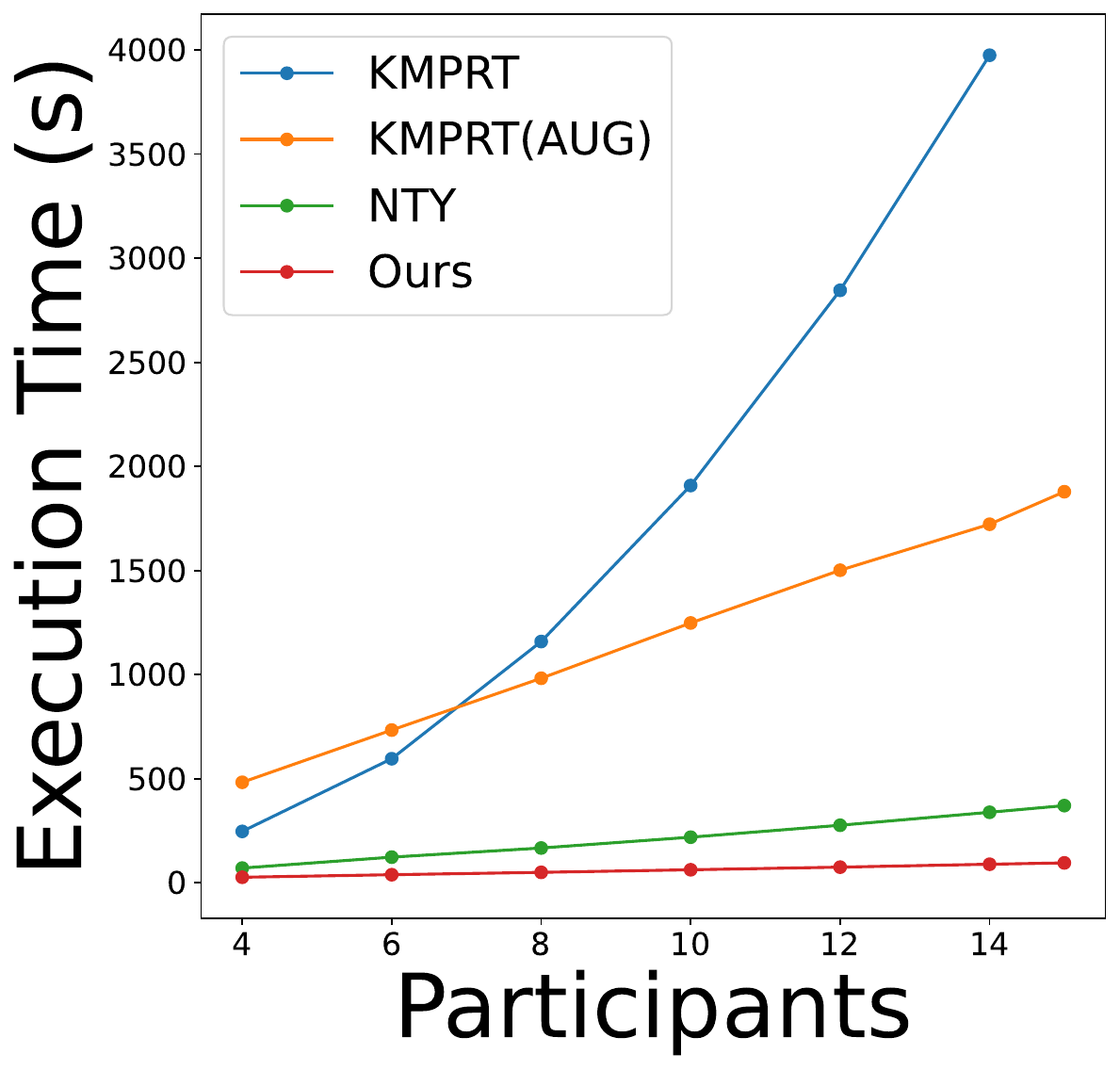}}
    \caption{Comparison of MPSI protocols under different data set sizes in WAN setting}
  \label{fig:wan}
\end{figure*}

\subsection{Theoretical analysis}
\label{subsec:parameter}
\subsubsection{Choice of $m, w$} The parameter $m,w$ in our protocol satisfy that if $F$ is a random function and $H_1(x)$ is different for each $x \in X_1\cup \dots \cup X_n$, then for $x\in X_n \setminus I$ and $v = F_k(H_1(x))$, there are at least $\lambda$ 1's in  $D^1_1[v[1]],\dots,D^1_w[v[w]]$ with all but negligible probability. 
This is to ensure that in the final step of the protocol, $P_1$ cannot brute-force information beyond obtaining the intersection from the OPRF values. In other words, the OPRF values for any item in the $P_n$'s input set which is not in the intersection $I$ remain pseudorandom to $P_1$. This is achieved because of the correlation robustness property of $H_2$. We will determine the value of $m$ (which is typically $N$ in this protocol) and then compute $w$ as follows.
The column $D_i$ was initialized as $1^m$, then for each $x\in X_1$, $P_1$ computes $v = F_k(H_1(x))$ and sets $D_i[v[i]] = 0$. Since $F$ is a random function and $H_1(x)$ is different for $x\in X_1$, $v$ is random and independent for $x\in X_1$ which means that the probability $\mathrm{Pr}[D_i^1[j] = 1] $ is same for all $j\in[m]$. especially, $\mathrm{Pr}[D_i^1[j] = 1] = (1- \frac{1}{m})^{n_1} $. Let $p = (1- \frac{1}{m})^{n_1}$. For any $x\in X_n \setminus I$, the probability that there are $d$ 1's in $D^1_1[v[1]],\dots,D^1_w[v[w]]$ is 
$$ \left( \begin{array}{c}
     w \\ d
\end{array} \right) p^d(1-p)^{w - d},$$
since $\mathrm{Pr}[D_i^1[v[i]] = 1] = p $ is independent for all $i \in [w]$. Then
for $x\in X_n \setminus I$, the proper $w$ that there are at least $d$ 1's with all but negligible probability can be computed by the union bound:
$$ N\cdot\sum_{k=0}^{d-1} \left( \begin{array}{c}
     w \\ k
\end{array} \right) p^k(1-p)^{w - k} \leq \negl[\sigma].$$

\subsubsection{Choice of $\ell_1$ and $\ell_2$} The parameters $\ell_1$ and $\ell_2$ are respectively the output lengths of hash functions $H_1$ and $H_2$. We need to set $\ell_1 = 2\lambda$ to guarantee collision resistance against the birthday attack and set $\ell_2 = \sigma + 2\log(N)$ to guarantee that the probability of collision in the MPSI protocol (Step \ref{h2}) is negligible for a semi-honest model, similarly as in \cite{pinkas2019spot,chase2020private}.

\subsubsection{Complexity analysis}
\label{subsec:complexity analysis}

In our MPSI protocol, we rely solely on lightweight cryptographic tools such at OT extension, hash functions, AES and bitwise operations. Without loss of generality, we consider the upper bound on the input set sizes of the parties as $N$\footnote{Typically, during testing, we assume that the size of the sets of all participants is $N$. Participants with sets smaller than $N$ can expand their set sizes to $N$ by adding random data.} and set $m=N$ as in \cite{chase2020private}. For the fixed $m,N$ and security parameter $\lambda$, we can compute $w, \ell_2$ using the method mentioned in Section \ref{subsec:parameter}. 

We denote party $P_1$ as the leader and party $P_2,\dots,P_{n-1}$ as the assistant. $P_1$ generates matrix $A$ and $B$ which cost linear complexity in $N$. Then $P_1$ runs $w$ OTs with $P_2$ that cost linear complexity in $N$. In this step, we significantly reduce communication and computational overhead by adopting the random OT scheme illustrated in Figure \ref{fig:ROT1}. 
In the final intersection calculation stage, $P_1$ only needs to perform hash functions and bitwise operation with a complexity of $N$.
For party $P_2,\dots,P_{n-1}$, they first act as the OT receiver and run OT to get the matrix $C^i$ for $i\in\{2,\dots,n-1 \}$. 
Then they generate matrix $D^i$ and run $w$ OTs with $P_{i+1}$ and use random OT in Figure \ref{fig:ROT2} to decrease communication and computational overhead which cost linear complexity in $N$. $P_n$ finally compute the OPRF value by hash functions and bitwise operation. 

Our protocol is currently the only one under the semi-honest model where the leader's and assistant's communication and computational complexities are equal. 
In traditional protocols, the leader typically bears a significant burden of communication and computation, which can impose considerable bandwidth and computational pressure on the leader.
The main reason for this situation is that most existing protocols predominantly employ mesh and star topology, which contribute to the imbalance in communication and computation workload, primarily borne by the leader. 
Kavousi et al. in \cite{kavousi2021efficient} proposed a protocol based on a star + wheel structure. In their protocol, all parties communicate with the leader once in the first step, followed by the protocol running in a wheel-like structure. 
Our protocol is a simple ring structure, where each participant communicates once with its adjacent party. This allows our protocol's communication and computational complexity for the leader and assistant to depend solely on the size of the input set, independent of the number of parties.

In table \ref{tab:complexity}, we compared the communication and computational complexity of our protocol with several currently efficient protocols \cite{kolesnikov2017practical,kavousi2021efficient,nevo2021simple,vos2022fast}. The criteria we selected for comparison include two points: first, the protocol needs to provide up to $n-1$ collusion resistance similar to our protocol, and second, the protocol needs to demonstrate high operational efficiency.

\begin{table*}[htb]
\centering
\caption{Running time (in millisecond) the MPSI protocols over large data sets}
\label{tab:bigdataLAN+WAN}
\renewcommand\arraystretch{1.2}
\setlength{\tabcolsep}{1mm}{
\begin{tabular}{|c|c|cccccccccc|}
\hline
\multirow{3}{*}{Set size} & \multirow{3}{*}{Protocol}  & \multicolumn{10}{c|}{Total running time for different parties(ms)}  \\ \cline{3-12} 
 & & \multicolumn{5}{c|}{LAN (5000M, 0.2 ms RTT)}  & \multicolumn{5}{c|}{WAN (200M, 2 ms RTT)} \\ \cline{3-12} 
 & & \multicolumn{1}{c|}{$|\mathcal{C}|=16$} & \multicolumn{1}{c|}{$|\mathcal{C}|=17$} & \multicolumn{1}{c|}{$|\mathcal{C}|=18$} & \multicolumn{1}{c|}{$|\mathcal{C}|=19$} & \multicolumn{1}{c|}{$|\mathcal{C}|=20$} & \multicolumn{1}{c|}{$|\mathcal{C}|=16$} & \multicolumn{1}{c|}{$|\mathcal{C}|=17$} & \multicolumn{1}{c|}{$|\mathcal{C}|=18$} & \multicolumn{1}{c|}{$|\mathcal{C}|=19$} & \multicolumn{1}{c|}{$|\mathcal{C}|=20$} \\ \hline
\multirow{2}{*}{$2^{20}$} & NTY~\cite{nevo2021simple} & \multicolumn{1}{c|}{43076}  & \multicolumn{1}{c|}{46154}  & \multicolumn{1}{c|}{49192}  & \multicolumn{1}{c|}{50213}  & \multicolumn{1}{c|}{55077}  & \multicolumn{1}{c|}{383345}  & \multicolumn{1}{c|}{409400}  & \multicolumn{1}{c|}{434396}  & \multicolumn{1}{c|}{463339}  & 492568  \\ \cline{2-12} 
 & Ours  & \multicolumn{1}{c|}{15094}  & \multicolumn{1}{c|}{15887}  & \multicolumn{1}{c|}{16185}  & \multicolumn{1}{c|}{15940}  & \multicolumn{1}{c|}{17154}  & \multicolumn{1}{c|}{101832}  & \multicolumn{1}{c|}{108228}  & \multicolumn{1}{c|}{113779}  & \multicolumn{1}{c|}{119942}  & 125969  \\ \hline
\multirow{2}{*}{$2^{21}$} & NTY~\cite{nevo2021simple} & \multicolumn{1}{c|}{91441}  & \multicolumn{1}{c|}{92183}  & \multicolumn{1}{c|}{99581}  & \multicolumn{1}{c|}{109904} & \multicolumn{1}{c|}{107423} & \multicolumn{1}{c|}{837114}  & \multicolumn{1}{c|}{902238}  & \multicolumn{1}{c|}{946369}  & \multicolumn{1}{c|}{998135}  & 1069145 \\ \cline{2-12} 
 & Ours  & \multicolumn{1}{c|}{29694}  & \multicolumn{1}{c|}{33125}  & \multicolumn{1}{c|}{30514}  & \multicolumn{1}{c|}{29988}  & \multicolumn{1}{c|}{31427}  & \multicolumn{1}{c|}{200637}  & \multicolumn{1}{c|}{213479}  & \multicolumn{1}{c|}{228281}  & \multicolumn{1}{c|}{238873}  & 253829  \\ \hline
\multirow{2}{*}{$2^{22}$} & NTY~\cite{nevo2021simple} & \multicolumn{1}{c|}{-} & \multicolumn{1}{c|}{-} & \multicolumn{1}{c|}{-} & \multicolumn{1}{c|}{-} & \multicolumn{1}{c|}{-} & \multicolumn{1}{c|}{-}  & \multicolumn{1}{c|}{-}  & \multicolumn{1}{c|}{-}  & \multicolumn{1}{c|}{-}  & -  \\ \cline{2-12} 
 & Ours  & \multicolumn{1}{c|}{64822}  & \multicolumn{1}{c|}{67968}  & \multicolumn{1}{c|}{66434}  & \multicolumn{1}{c|}{69497}  & \multicolumn{1}{c|}{67659}  & \multicolumn{1}{c|}{406194}  & \multicolumn{1}{c|}{430619}  & \multicolumn{1}{c|}{457885}  & \multicolumn{1}{c|}{485455}  & 511640  \\ \hline
\multirow{2}{*}{$2^{23}$} & NTY~\cite{nevo2021simple} & \multicolumn{1}{c|}{-} & \multicolumn{1}{c|}{-} & \multicolumn{1}{c|}{-} & \multicolumn{1}{c|}{-} & \multicolumn{1}{c|}{-} & \multicolumn{1}{c|}{-}  & \multicolumn{1}{c|}{-}  & \multicolumn{1}{c|}{-}  & \multicolumn{1}{c|}{-}  & -  \\ \cline{2-12} 
 & Ours  & \multicolumn{1}{c|}{138428} & \multicolumn{1}{c|}{137397} & \multicolumn{1}{c|}{136229} & \multicolumn{1}{c|}{138187} & \multicolumn{1}{c|}{138208} & \multicolumn{1}{c|}{809523}  & \multicolumn{1}{c|}{862368}  & \multicolumn{1}{c|}{916236}  & \multicolumn{1}{c|}{970604}  & 1022665 \\ \hline
\multirow{2}{*}{$2^{24}$} & NTY~\cite{nevo2021simple} & \multicolumn{1}{c|}{-} & \multicolumn{1}{c|}{-} & \multicolumn{1}{c|}{-} & \multicolumn{1}{c|}{-} & \multicolumn{1}{c|}{-} & \multicolumn{1}{c|}{-}  & \multicolumn{1}{c|}{-}  & \multicolumn{1}{c|}{-}  & \multicolumn{1}{c|}{-}  & -  \\ \cline{2-12} 
 & Ours  & \multicolumn{1}{c|}{282774} & \multicolumn{1}{c|}{281785} & \multicolumn{1}{c|}{288982} & \multicolumn{1}{c|}{281654} & \multicolumn{1}{c|}{295784} & \multicolumn{1}{c|}{1624301} & \multicolumn{1}{c|}{1721287} & \multicolumn{1}{c|}{1842714} & \multicolumn{1}{c|}{1949102} & 2053595 \\ \hline
\end{tabular}}
\end{table*}

\begin{table}[htb]
\centering
\caption{Communication cost (in MB) of MPSI protocols.}
\label{tab:COM}
\renewcommand\arraystretch{1.2}
\setlength{\tabcolsep}{1mm}{
\begin{tabular}{|c|c|ccc|}
\hline
\multirow{2}{*}{\textbf{Set size}}  & \multirow{2}{*}{\textbf{Protocol}} & \multicolumn{3}{c|}{\textbf{Communication cost (MB)}} \\ \cline{3-5} 
  &  & \multicolumn{1}{c|}{$|\mathcal{C}|=4$} & \multicolumn{1}{c|}{$|\mathcal{C}|=10$}  & $|\mathcal{C}|=15$  \\ \hline
\multirow{4}{*}{$2^{12}$} & KMPRT~\cite{kolesnikov2017practical} & \multicolumn{1}{c|}{19.68} & \multicolumn{1}{c|}{147.6} & 344.4 \\ \cline{2-5} 
  & KMPRT(AUG)~\cite{kolesnikov2017practical} & \multicolumn{1}{c|}{7.36}  & \multicolumn{1}{c|}{18.4}  & 27.75 \\ \cline{2-5} 
  & NTY~\cite{nevo2021simple} & \multicolumn{1}{c|}{7.36}  & \multicolumn{1}{c|}{18.4}  & 27.75 \\ \cline{2-5} 
  & Ours & \multicolumn{1}{c|}{\textbf{1.6}} & \multicolumn{1}{c|}{\textbf{5.20}}  & \textbf{8.24}  \\ \hline
\multirow{4}{*}{$2^{16}$} & KMPRT~\cite{kolesnikov2017practical}  & \multicolumn{1}{c|}{311.2} & \multicolumn{1}{c|}{2334.1}  & 5446.35 \\ \cline{2-5} 
  & KMPRT(AUG)~\cite{kolesnikov2017practical} & \multicolumn{1}{c|}{126.76}  & \multicolumn{1}{c|}{316.9} & 475.35  \\ \cline{2-5} 
  & NTY~\cite{nevo2021simple}  & \multicolumn{1}{c|}{126.76}  & \multicolumn{1}{c|}{316.9} & 475.35  \\ \cline{2-5} 
  & Ours & \multicolumn{1}{c|}{\textbf{21.53}} & \multicolumn{1}{c|}{\textbf{79.52}} & \textbf{126.4}  \\ \hline
\multirow{4}{*}{$2^{20}$} & KMPRT~\cite{kolesnikov2017practical}  & \multicolumn{1}{c|}{5608}  & \multicolumn{1}{c|}{42080} & 98205 \\ \cline{2-5} 
  & KMPRT(AUG)~\cite{kolesnikov2017practical}  & \multicolumn{1}{c|}{2225.56} & \multicolumn{1}{c|}{5563.9}  & 8345.85 \\ \cline{2-5} 
  & NTY~\cite{nevo2021simple}  & \multicolumn{1}{c|}{2225.56} & \multicolumn{1}{c|}{5563.9}  & 8345.85 \\ \cline{2-5} 
  & Ours & \multicolumn{1}{c|}{\textbf{357.84}}  & \multicolumn{1}{c|}{\textbf{1253.42}} & \textbf{2106.23} \\ \hline
\end{tabular}}
\end{table}

\subsection{Experimental analysis}
\subsubsection{Experimental setup} Our MPSI was implemented based on libOTe~(https://github.com/osu-crypto/libOTe) using C++-11. We compared our MPSI with the state-of-the-art ones, i.e., the KMPRT protocol~\cite{kolesnikov2017practical} and the NTY protocol~\cite{nevo2021simple}. All protocols in comparison were run on Ubuntu 22.04 servers with $2\times$ Intel(R) Xeon(R) 2.10 GHz CPU and 256GB RAM. The input length was set to be 128 bits, the statistical security parameter is $\kappa = 40$, and the computational security parameter is $\lambda = 128$. All protocols were tested in both the local area network (LAN) and wide area network (WAN) settings, the LAN setting has 0.2 milliseconds round-trip latency and 5000 Mbps network bandwidth, and the WAN setting has 2 milliseconds round-trip latency and 200 Mbps network bandwidth. The dataset sizes ranges from $2^{12}$ to $2^{20}$, the number of parties ranges from $4$ to $15$.

\subsubsection{Comparison on total running time} Table~\ref{tab:LAN+WAN} shows the total running time of the related protocols. It is evident that as the dataset size expands, our MPSI protocol not only maintains a shorter overall execution time compared to other protocols but also increasingly outperforms them. This improvement is primarily attributed to the substantial reduction in communication overhead achieved by our MPSI scheme as the dataset grows, as clearly demonstrated by comparing experimental results in WAN and LAN environments. Fig.~\ref{fig:lan} and Fig.~\ref{fig:wan} further illustrates the comparative trends of the MPSI protocol's execution time as the number of participants increases across various dataset sizes in the LAN and WAN settings, respectively. It is evident that the execution time of our MPSI protocol grows linearly at a very slow pace. Particularly in WAN settings, our protocol yields the best results compared to all others.

\subsubsection{Comparison over large data sets} For data sets over the size of $2^{20}$, and for the number of participants, values greater than or equal to 16. We only selected NTY~\cite{nevo2021simple} and our own protocol for comparison. The NTY protocols was unable to run due to insufficient memory when the data volume exceeded 21. The experimental results are shown in Table \ref{tab:bigdataLAN+WAN}.

\subsubsection{Comparison on communication overhead} Table~\ref{tab:COM} demonstrates that our scheme substantially reduces communication overhead. Specifically, with a dataset size of $2^{20}$ and $15$ participants, our protocol achieves a $97.8\%$ reduction compared to the KMPRT protocol and a $74.8\%$ reduction compared to the NTY protocol.

\section{Conclusion}\label{sec:conclusion}
In this paper, we propose an efficient Multi-Party Private Set Intersection (MPSI) protocol aimed at addressing the bottlenecks in computational and communication overhead in existing MPSI protocols. By constructing a new multi-party sequential oblivious pseudorandom function (MP-SOPRF), our protocol not only enhances computational efficiency but also significantly reduces the communication cost required for multi-party interactions. Our protocol is the first MPSI protocol based solely on a ring topology. Experiment results show that our MPSI achieves an approximate 2.87$\times$ acceleration in total running time and achieves a $74.8\%$ reduction about communication, which is crucial for large-scale data set applications.

\ifCLASSOPTIONcompsoc
\else
\fi




%


\end{document}